\documentclass[11pt,a4paper]{article}

\usepackage[utf8]{inputenc}
\usepackage[T1]{fontenc}
\usepackage[a4paper,margin=1in]{geometry}
\usepackage{amsmath,amssymb,amsthm,mathtools}
\usepackage{array,multirow,tabularx}
\usepackage{caption}
\usepackage{xcolor}
\usepackage{algorithm}
\usepackage{algpseudocode}
\usepackage{booktabs}
\usepackage{graphicx}
\usepackage{hyperref}
\usepackage{natbib}
\usepackage{bm}
\usepackage{setspace}
\usepackage{cleveref}
\usepackage[section]{placeins}
\usepackage{enumitem}

\newcommand{\QN}{\mathbb{Q}^N}
\newcommand{\QR}{\mathbb{Q}^R}
\newcommand{\QC}{\mathbb{Q}^C}
\newcommand{\sigvec}[1]{{\bm{\sigma}^{#1}}}
\newcommand{\sigI}{\sigma_I}
\newcommand{\sigJ}{\sigma_J}
\newcommand{\BE}{\mathrm{BE}}
\newcommand{\dd}{\,\mathrm{d}}
\newcommand{\Avec}[1]{\bm{A}^{#1}}
\newcommand{\Ft}{\mathcal{F}_t}

\theoremstyle{plain}
\newtheorem{proposition}{Proposition}
\theoremstyle{definition}
\newtheorem{definition}[proposition]{Definition}
\theoremstyle{remark}
\newtheorem{remark}[proposition]{Remark}

\title{Three-Currency HJM for Brazilian Credit Markets}
\author{Raphael R. R. Coelho\thanks{Independent researcher. Email: \href{mailto:raphaelrrcoelho@gmail.com}{raphaelrrcoelho@gmail.com}.}}
\date{\today}

\begin{document}

\onehalfspacing
\maketitle

\begin{abstract}
\noindent This paper develops a three-currency Heath-Jarrow-Morton framework in which corporate credit is treated as a separate economy, connected to the nominal and real economies through synthetic inflation and credit exchange rates. The framework produces a testable identity. Under joint no-arbitrage, the credit spread of an issuer expressed over the inflation-rate-indexed risk-free curve equals the same issuer's credit spread expressed over the nominal-rate-indexed risk-free curve plus the model-implied breakeven inflation forward at the same maturity. The identity holds within any single calibration of the framework. It is empirically falsifiable across two parallel corporate-bond segments of the same market: in a segmented market the two segments may price different corporate credit economies, and the gap between their implied corporate forwards measures the failure of the shared-credit-economy assumption.

\medskip

\noindent Applied to Brazilian debenture markets, the framework delivers a sharp empirical finding. Fifteen large issuers placed paper in both the CDI-indexed general-purpose segment and the IPCA-indexed infrastructure segment between January 2021 and February 2026. The within-issuer triangle residual at the 3-year tenor averages 640 basis points, with cross-sectional standard deviation of 26 basis points across the 15 issuer means, and remains stable through both the 2021-2023 BCB tightening cycle and the 2024-2026 easing phase. A retail post-tax indifference benchmark anchored on \emph{Lei}~12.431 closes the bulk of the residual. The remainder is consistent with institutional participation on the CDI side, contractual asymmetries between debentures with different use-of-proceeds restrictions, and segment-specific liquidity gaps.
\end{abstract}

\medskip


\medskip


\bigskip

\section{Introduction}\label{sec:intro}

\citet{heath1992bond} showed that arbitrage-free term-structure dynamics impose a drift restriction on forward rates within a single curve. The restriction determines the drift of each forward rate from the volatility structure once a numeraire is fixed, and it remains the workhorse condition behind most of the term-structure literature that followed. Less studied is the parallel question for an investor who simultaneously observes several term structures of different indexations on the same set of underlying issuers. When a single legal entity raises debt in a nominal-rate-indexed segment and in an inflation-rate-indexed segment of the same fixed-income market, what no-arbitrage restriction must its two credit-spread curves satisfy? The within-curve drift condition is silent on the question.

The natural setting in which to ask it is a Heath-Jarrow-Morton economy with three currencies. \citet{jarrow2003hjm} treat the nominal and real economies as separate economies linked through inflation as a synthetic exchange rate. The construction in this paper adds corporate credit as a third economy connected to the nominal economy through a credit exchange rate, defined as the discounted ratio of a corporate-credit bank account and the nominal-risk-free bank account. The three pricing measures are linked through their numeraire ratios, and joint absence of arbitrage forces an algebraic identity between three observable forward objects: the issuer's credit spread over the nominal-indexed risk-free curve, the same issuer's credit spread over the inflation-indexed risk-free curve, and the model-implied breakeven inflation forward at the same maturity.

The identity says the inflation-indexed credit spread equals the nominal-indexed credit spread plus the breakeven inflation forward, provided both segments price the same underlying corporate-credit economy. It is the cross-curve analog of the HJM drift condition: it does not constrain the dynamics of any single curve, but it constrains the joint level of two corporate spread curves observed under different risk-free bases.

The identity is testable. Each segment defines its own observed corporate forward, equal to its risk-free forward plus the segment's observed credit spread. The framework's shared-credit-economy null is that these two implied corporate forwards coincide. The empirical content of the framework, in a market with two parallel corporate segments, is whether the null holds. If it does, both segments price the same credit economy and the framework's three-currency structure characterises that pricing. If it fails, the gap between the two implied corporate forwards measures the failure of the shared-credit-economy assumption, and the explanation lies outside the framework, in the institutional mechanism that separates the two clienteles.

Brazilian debenture markets are a setting in which the test is sharply informative. Two parallel segments exist with partly overlapping but institutionally distinct issuer universes: general-purpose debentures indexed to CDI and infrastructure debentures indexed to IPCA. Both segments trade through JGP Idex constituent baskets. \emph{Lei}~12.431 exempts individual investors from the 15 percent withholding tax otherwise applied to debenture coupons under \emph{Lei}~11.033, conditional on the debenture qualifying as infrastructure under the statute. The exemption produces a marginal-investor wedge: retail investors are tax-advantaged on the IPCA side and tax-disadvantaged on the CDI side, while institutional investors, who dominate the CDI segment, are not subject to the wedge. The wedge breaks the shared-credit-economy assumption in a known and quantifiable direction.

For 15 large issuers active in both segments between 2021 and 2026, the within-issuer triangle residual at the 3-year tenor averages 640 basis points, with cross-sectional standard deviation of 26 basis points across the 15 issuer means. The narrow cross-sectional dispersion indicates that the gap is a property of the segments rather than of issuer-specific credit risk. A retail post-tax indifference benchmark anchored on \emph{Lei}~12.431 accounts for the bulk of the residual. The remainder is consistent with institutional participation on the CDI side and with structural and liquidity asymmetries between general-purpose and use-of-proceeds-restricted infrastructure debentures.

The contribution sits between two strands of the literature. The first is multi-currency term-structure modelling. \citet{jarrow2003hjm} provide the nominal-real foundation. Credit extensions of HJM developed in \citet{schonbucher2000libor}, \citet{eberlein2003defaultable}, and \citet{bielecki2002credit} treat credit within a single economy and study default-time intensity rather than parallel-curve spreads; the intensity-based framework of \citet{duffie1999modeling} provides the closest precedent for the reduced-form treatment of corporate credit adopted here. Multi-curve frameworks following \citet{mercurio2010interest} address post-crisis basis spreads between OIS, LIBOR, and tenor-specific discount curves, but the curves they relate share an indexation. The closest precedent in spirit is the inflation-linked literature on real-nominal arbitrage in sovereign markets, which uses a two-curve breakeven identity to characterise no-arbitrage between government nominal and inflation-indexed bonds. The contribution here is the three-curve restriction that emerges when corporate credit is added as a separate economy alongside nominal and real, and the recognition that the resulting identity becomes empirically informative when applied to two parallel corporate segments of the same fixed-income market.

The second strand is empirical clientele effects in fixed income. Tax-induced retail clientele effects have a long history in U.S.\ municipal bonds, where the marginal-investor tax rate is identified from the spread between municipal and taxable corporate yields. The most directly related Brazilian literature studies the fiscal efficiency of \emph{Lei}~12.431 from the issuer's side. \citet{pereira2021nota} estimate the fraction of foregone federal tax revenue that transfers to project issuers through reduced funding costs, finding pass-through that is partial and varies with debenture duration and rating. \citet{brand2022debentures} tests whether the fiscal benefit is appropriated by issuers under a no-arbitrage condition on the secondary-market spread between incentivised and non-incentivised debentures. Both studies focus on the supply side of the wedge. The within-issuer cross-segment test developed here measures the wedge from the investor's pricing under joint no-arbitrage in a three-currency framework, delivering a precise quantitative magnitude (640 basis points at the 3-year tenor, with cross-sectional standard deviation of 26 basis points across 15 issuers). The Brazilian credit-spread literature has also documented level and time-series determinants of debenture spreads under the local term-structure setting \citep{sheng2005determinantes,varga2009teste}; the within-issuer cross-segment test exploiting the parallel CDI- and IPCA-indexed corporate markets has not been reported. The Brazilian setting differs from the U.S.\ muni precedent in two respects. The exemption is type-specific (infrastructure debentures only) rather than issuer-specific (municipal issuers in the U.S.), and the empirical identification is within-issuer rather than across heterogeneous universes. Within-issuer identification removes rating, sector, and duration composition from the comparison, leaving the segmentation channel as the residual to interpret.

\section{The three-currency framework}\label{sec:framework}

The framework treats nominal, real, and corporate-credit instruments as three connected economies in the sense of \citet{jarrow2003hjm}. Each economy carries its own pricing measure and bank-account numeraire. Two synthetic exchange rates, one for inflation and one for credit, link the three measures, and joint absence of arbitrage pins down both the within-curve drift restrictions and the cross-economy relations between forward-rate dynamics.

\begin{remark}[Inner-product convention]
\label{rem:inner}
For volatility-loading vectors $x,y\in\mathbb{R}^m$, the symbol $x\cdot y$ throughout the paper denotes the $\bm{\rho}$-weighted form $x^{\top}\bm{\rho}\,y$, equal to the instantaneous quadratic covariation $d\langle x\cdot W^N, y\cdot W^N\rangle_t/dt$, Under the canonical-basis volatility decomposition with PC-aligned blocks introduced in \Cref{subsec:shared_factor,def:factor_blocks}, $\bm{\rho}$ reduces to the identity within each block by construction (PCA factor scores are orthogonal in time), and the standard Euclidean inner product applies for vectors supported on a single block. Cross-block entries of $\bm{\rho}$ are estimated empirically (\Cref{sec:calib_corr}) and enter formulas only through inner products of vectors supported on different blocks.
\end{remark}

\subsection{Economies and numeraires}\label{subsec:economies}

Let $(\Omega, \mathcal{F}, (\mathcal{F}_t)_{t \ge 0}, \QN)$ be a filtered probability space carrying a $d$-dimensional Brownian motion $W^N$ under the nominal pricing measure $\QN$. The nominal economy is a Heath-Jarrow-Morton economy: a family of forward rates $\{f^N(t,T) : T \ge t\}$ satisfies $r^N(t) = f^N(t,t)$, where $r^N$ is the instantaneous nominal risk-free rate, and evolves under $\QN$ as
\begin{equation}\label{eq:fN_dyn}
df^N(t,T) = \alpha^N(t,T)\,dt + \sigvec{N}(t,T)^{\top} dW^N(t).
\end{equation}
The drift $\alpha^N$ is pinned by the HJM no-arbitrage condition to $\sigvec{N}(t,T)^{\top} \int_t^T \sigvec{N}(t,u)\,du$. The nominal bank account is $B^N_t = \exp\!\bigl(\int_0^t r^N(u)\,du\bigr)$.

The real economy is defined symmetrically. Under its pricing measure $\QR$, a real bank account $B^R_t = \exp\!\bigl(\int_0^t r^R(u)\,du\bigr)$ accrues at the instantaneous real risk-free rate $r^R(t) = f^R(t,t)$, and the real forward curve evolves as
\begin{equation}\label{eq:fR_dyn}
df^R(t,T) = \alpha^R(t,T)\,dt + \sigvec{R}(t,T)^{\top} dW^R(t),
\end{equation}
with $\alpha^R$ determined by the within-curve HJM restriction under $\QR$.

The corporate-credit economy follows the same pattern. Let $s^C(t,t)$ denote the instantaneous corporate-credit spread of a representative issuer above the nominal short rate; the credit short rate is $r^C(t) = r^N(t) + s^C(t,t)$. The credit bank account $B^C_t = \exp\!\bigl(\int_0^t r^C(u)\,du\bigr)$ accrues at the rate at which a continuously-rolled, instantaneously-resetting corporate-credit position would compound. Under the credit pricing measure $\QC$, the corporate-credit forward curve $\{f^C(t,T)\}$ evolves as
\begin{equation}\label{eq:fC_dyn}
df^C(t,T) = \alpha^C(t,T)\,dt + \sigvec{C}(t,T)^{\top} dW^C(t),
\end{equation}
with the corresponding HJM drift restriction under $\QC$.

\subsection{Exchange rates}\label{subsec:fx}

Two synthetic exchange rates connect the three economies. The inflation exchange rate $I_t$ is the level of the consumer price index at $t$, normalised so that $I_0 = 1$. Absence of arbitrage between nominal and inflation-indexed instruments requires that $I_t B^R_t / B^N_t$ be a $\QN$-martingale, which is equivalent to the dynamics
\begin{equation}\label{eq:I_dyn}
\frac{dI_t}{I_t} = \bigl[r^N(t) - r^R(t)\bigr]\,dt + \sigvec{I}(t)^{\top} dW^N(t),
\end{equation}
with inflation-FX volatility $\sigvec{I}(t)$ \citep{jarrow2003hjm}. The drift $r^N - r^R$ is the model-implied breakeven inflation rate.

The credit exchange rate $J_t$ is the analogous object connecting the nominal and corporate-credit economies. Define $J_t = B^N_t / B^C_t$, the cumulative discount of a corporate-credit-numeraire position relative to a nominal-risk-free position. Under $\QN$, $J_t$ satisfies
\begin{equation}\label{eq:J_dyn}
\frac{dJ_t}{J_t} = -s^C(t,t)\,dt + \sigvec{J}(t)^{\top} dW^N(t),
\end{equation}
with credit-FX volatility $\sigvec{J}(t)$. A position long $J$ loses the instantaneous corporate-credit spread per unit time, by construction.

The volatility $\sigvec{J}(t)$ governs the dynamics of a synthetic instrument that is not directly traded. It is not separately identified from histories of the nominal, real, and corporate forward curves alone: its marginal contribution to a likelihood function over those three curves is degenerate against the contributions of $\sigvec{N}, \sigvec{R}, \sigvec{C}$ under any finite-factor decomposition consistent with the data. Pricing of credit-FX-sensitive instruments such as quanto credit derivatives requires fixing $\sigvec{J}$ by economic convention or by external instruments. The empirical content of the triangle identity derived in Section~\ref{sec:restriction} does not depend on $\sigvec{J}$.

\subsection{Change of measure}\label{subsec:girsanov}

Girsanov's theorem relates the three pricing measures through the exchange-rate dynamics. The Radon-Nikodym derivative from $\QN$ to $\QR$ is
\begin{equation}\label{eq:RN_NR}
\left.\frac{d\QR}{d\QN}\right|_{\mathcal{F}_t} = \mathcal{E}\!\left(\int_0^t \sigvec{I}(u)^{\top}\,dW^N(u)\right),
\end{equation}
where $\mathcal{E}(\cdot)$ denotes the stochastic exponential. The corresponding Brownian motion under $\QR$ is $W^R(t) = W^N(t) - \int_0^t \sigvec{I}(u)\,du$. The Radon-Nikodym derivative from $\QN$ to $\QC$ takes the analogous form with $\sigvec{J}$ in place of $\sigvec{I}$, yielding $W^C(t) = W^N(t) - \int_0^t \sigvec{J}(u)\,du$. The change of measure between $\QR$ and $\QC$ is the composition, with Brownian-motion shift $\int_0^t [\sigvec{J}(u) - \sigvec{I}(u)]\,du$.

Drift corrections under change of measure produce the cross-economy relations between forward-rate dynamics. A process whose $\QN$ dynamics include $\sigvec{X}(t)^{\top} dW^N(t)$ acquires a quadratic-covariance correction $-\sigvec{X}(t)^{\top} \sigvec{I}(t)\,dt$ when its dynamics are rewritten under $\QR$, and $-\sigvec{X}(t)^{\top} \sigvec{J}(t)\,dt$ under $\QC$. These corrections close the system of within-curve drift restrictions across the three economies. Appendix~\ref{app:hjm} contains the explicit bond-price dynamics under each measure and the resulting joint drift constraints on $\sigvec{N}, \sigvec{R}, \sigvec{C}$.

\subsection{Spread definitions}\label{subsec:spreads}

A corporate issuer that raises debt under two parallel index bases admits two conventions for its credit spread. Let $f^C(t,T)$ denote the issuer's corporate-credit forward curve under the framework. The credit spread expressed over the nominal-rate-indexed risk-free curve is
\begin{equation}\label{eq:sN_def}
s^{N}(t,T) \equiv f^C(t,T) - f^N(t,T),
\end{equation}
and the credit spread expressed over the inflation-rate-indexed risk-free curve is
\begin{equation}\label{eq:sR_def}
s^{R}(t,T) \equiv f^C(t,T) - f^R(t,T).
\end{equation}
The two spreads share the same corporate forward $f^C$ and differ only in the risk-free basis against which the credit is measured. The breakeven inflation forward is the standard differential
\begin{equation}\label{eq:BE_def}
\BE(t,T) \equiv f^N(t,T) - f^R(t,T),
\end{equation}
which is the model-implied breakeven inflation rate at maturity $T$.

In the Brazilian instantiation developed in Section~\ref{sec:institutional}, the nominal-rate-indexed risk-free curve $f^N$ is constructed from DI futures, the standard market for Brazilian nominal-rate term structure, and the inflation-rate-indexed risk-free curve $f^R$ is constructed from NTN-B prices, Brazilian Treasury notes whose principal is indexed to IPCA. The credit spread $s^N$ is identified with the spread of corporate debentures indexed to CDI (denoted $s^{\mathrm{CDI}}$), and the credit spread $s^R$ is identified with the spread of corporate debentures indexed to IPCA (denoted $s^{\mathrm{IPCA}}$).

\section{The triangle restriction}\label{sec:restriction}

\begin{proposition}[Triangle identity]\label{prop:triangle}
Under the framework of Section~\ref{sec:framework}, the credit-spread conventions \eqref{eq:sN_def} and \eqref{eq:sR_def} satisfy, identically in $t$ and $T$,
\begin{equation}\label{eq:triangle}
s^R(t,T) \;=\; s^N(t,T) + \BE(t,T).
\end{equation}
\end{proposition}

\begin{proof}
Subtract \eqref{eq:sN_def} from \eqref{eq:sR_def}. The corporate forward $f^C$ cancels, leaving $f^N - f^R = \BE$.
\end{proof}

The identity is algebraic, not statistical. It follows from the definitional structure of the two spreads, both expressed relative to the same corporate forward $f^C$. Within any single calibration of the framework, the identity holds by construction. The framework's empirical content arises when the identity is applied across two parallel corporate-bond segments of a market, each segment generating its own observed corporate forward.

\subsection{From algebraic identity to empirical restriction}\label{subsec:empirical_content}

In a market with two parallel corporate-bond segments, two corporate forwards are observable. The nominal-indexed segment yields $f^{C,N}_{\mathrm{mkt}}(t,T) = f^N(t,T) + s^N_{\mathrm{mkt}}(t,T)$ from its observed credit spread $s^N_{\mathrm{mkt}}$. The inflation-indexed segment yields $f^{C,R}_{\mathrm{mkt}}(t,T) = f^R(t,T) + s^R_{\mathrm{mkt}}(t,T)$ from its observed credit spread $s^R_{\mathrm{mkt}}$. The framework's shared-credit-economy null is that the two implied corporate forwards coincide,
\begin{equation}\label{eq:H0}
H_0: \quad f^{C,N}_{\mathrm{mkt}}(t,T) = f^{C,R}_{\mathrm{mkt}}(t,T) \quad \text{for all } t,\,T.
\end{equation}
Equivalently, define the within-issuer triangle residual
\begin{equation}\label{eq:Delta_def}
\Delta(t,T) \;\equiv\; s^R_{\mathrm{mkt}}(t,T) - s^N_{\mathrm{mkt}}(t,T) - \BE(t,T).
\end{equation}
The null is $\Delta \equiv 0$. The residual $\Delta$ measures the gap between the two implied corporate forwards in basis points.

The residual admits a transparent interpretation. If the two segments share the same underlying corporate-credit economy, the two implied corporate forwards coincide and $\Delta = 0$. If the two segments price different corporate economies, $\Delta$ measures the gap. Three classes of mechanisms can produce $\Delta \ne 0$: segmentation of marginal investors across the two segments by tax treatment, by investor type, or by institutional access; contractual differences between the segments, such as covenant packages, use-of-proceeds restrictions, or ring-fencing of cash flows, that produce different effective credit risks for the same legal entity; and liquidity differences that produce a discount on the less actively traded side. The test of $\Delta = 0$ does not separately identify these channels but rejects the shared-credit-economy assumption in their joint presence.

\subsection{Interpretation as a cross-curve no-arbitrage restriction}\label{subsec:interpretation}

The HJM literature derives drift restrictions on forward rates within a single curve, conditional on a chosen numeraire. The triangle identity \eqref{eq:triangle} is a cross-curve restriction of a different structure. It does not constrain the dynamics of any single forward rate. It constrains the joint level of two corporate spread curves measured against two different risk-free bases. The two are equivalent if and only if the framework's shared-credit-economy assumption holds.

A failure of $\Delta \equiv 0$ does not invalidate the HJM framework or the change-of-measure machinery. It rejects only the joint hypothesis that one corporate-credit economy describes both segments. The empirical violation of $\Delta = 0$ is informative when accompanied by a candidate mechanism that breaks the shared-credit-economy hypothesis in a known direction. Section~\ref{sec:institutional} describes the regulatory mechanism in the Brazilian setting that breaks the assumption: an income-tax exemption for individual investors in the inflation-indexed segment, sustained by a 2011 statute (Lei 12.431) and absent in the nominal-indexed segment.

\section{Brazilian debenture markets and data}\label{sec:institutional}

Brazilian fixed-income markets price three benchmark curves and three corporate-credit segments in parallel. The CDI overnight rate anchors the nominal curve through DI futures on B3; ANBIMA-mediated NTN-B prices anchor the real (IPCA-linked) curve; corporate debentures trade in two index-tracked families, JGP Idex-CDI and JGP Idex-Infra. The Idex-CDI Core universe tracks the spread of liquid, rated CDI-linked corporate debentures over the DI curve across a broad issuer base. The Idex-Infra Core universe tracks the spread of IPCA-linked infrastructure debentures over the NTN-B curve; these instruments are typically issued under \emph{Lei}~12.431.

\paragraph{Lei~12.431 and the retail tax exemption.}
Enacted in 2011 to encourage retail funding of long-tenor infrastructure investment, \emph{Lei}~12.431 exempts qualifying IPCA-linked infrastructure debentures from Brazilian individual-investor income tax on coupons and amortization. The exemption applies to securities meeting use-of-proceeds restrictions tied to specific infrastructure sectors (energy, sanitation, logistics, telecommunications) and to minimum-tenor and amortization-profile requirements set by inter-ministerial decree. Comparable long-tenor non-infrastructure CDI-linked corporate debentures held by individual investors longer than 720 days face the 15\% rate of \emph{Lei}~11.033. The two regimes coexist over our sample with no overlap in coverage: an infrastructure debenture is either qualifying under 12.431 (exempt) or it is not (subject to 11.033).

\paragraph{Dual-listed issuers.}
Many large Brazilian corporates with infrastructure assets fund themselves through both index families. Telecom operators, electricity utilities, sanitation concessionaires, and integrated logistics operators commonly place CDI-linked general-purpose debentures alongside IPCA-linked debentures issued under \emph{Lei}~12.431 against specific infrastructure projects. For each such issuer, the two index families produce two parallel spread quotes referenced to a common legal entity. In our sample, 15 issuers have at least 800 days of joint observation in both index families at the 3-year tenor; the within-issuer comparison restricted to these issuers is the headline test of Section~\ref{subsec:dual_listed}.

\paragraph{Marginal investor and clientele structure.}
The qualifying IPCA-linked infrastructure debenture market is, by construction of the exemption, populated disproportionately by individual investors holding either directly or through self-directed brokerage accounts. The CDI-linked corporate market is held more broadly across investment vehicles: credit funds, multimarket and pension funds, and corporate treasuries. Institutional ownership of \emph{Lei}~12.431 paper exists but is non-marginal at the level of the headline retail-incentive-driven flow. The distinction at the margin between the two clienteles is what generates the tax-driven price wedge tested below.

\subsection{Data and curve construction}\label{subsec:data}

The empirical test uses four time series aligned on common business dates: nominal forwards from DI futures, real forwards from NTN-B prices, and credit-spread curves from JGP Idex-CDI and Idex-Infra. This subsection describes each source and the alignment.

\subsubsection{Nominal curve: DI futures}
\label{subsec:data_di}

The nominal term structure is extracted from B3-listed DI1 (one-day interbank deposit) futures contracts \citep{b3_di1}. B3 defines DI1 as an exchange-traded futures contract on the average daily interbank-deposit rate between trade date and maturity, with all listed monthly maturities eligible for trading. From the strip of DI1 settlement prices, we bootstrap zero-coupon rates using the standard business-day compounding convention:
\begin{equation}\label{eq:di_bootstrap}
  P^N(t, T_i) = \frac{1}{\big(1 + y^N(t,T_i)\big)^{du(t,T_i)/252}},
\end{equation}
where $y^N(t,T_i)$ is the annualized yield and $du(t,T_i)$ the number of business days from $t$ to $T_i$ following the ANBIMA calendar. We convert to continuously compounded instantaneous forward rates via:
\begin{equation}\label{eq:fwd_from_zcb}
  f^N(t,T) = -\frac{\partial}{\partial T}\ln P^N(t,T),
\end{equation}
implemented numerically as a finite difference on the interpolated zero curve. Interpolation between DI1 maturities uses the flat-forward convention (piecewise-constant forward rates between nodes), which is the market standard and avoids introducing curvature artifacts.

The resulting dataset consists of daily nominal forward curves at pillar maturities $\tau \in \{0.25, 0.50, 1, 2, 3, 5, 7, 10\}$ years over 1{,}270 business days from 2021-01-04 to 2026-02-05 (after alignment with NTN-B and JGP series, below).

\subsubsection{Real curve: NTN-B government bonds}
\label{subsec:data_ntnb}

The real term structure is derived from NTN-B (Notas do Tesouro Nacional, S\'erie B) prices, which are IPCA-linked government bonds paying semiannual coupons of 6\% p.a.\ on the inflation-adjusted principal. We use ANBIMA indicative prices \citep{anbima_ntnb}, which reflect the median of broker-dealer quotes and are the standard source for marking-to-market in Brazilian institutional portfolios.

From the NTN-B yield-to-maturity quotes $\{y^{NTN\text{-}B}(t, T_i)\}$, we strip the coupon structure to obtain zero-coupon real rates via a bootstrap procedure that accounts for:
\begin{enumerate}[label=(\roman*)]
  \item The business-day compounding convention (same $du/252$ structure as DI).
  \item The semi-annual coupon schedule with IPCA indexation.
  \item The VNA (Valor Nominal Atualizado) accrual convention for the inflation-adjusted principal.
\end{enumerate}

Over the sample, the outstanding NTN-B universe comprises 9 to 14 maturities at any given date, with the longest currently extending to 2060. The stripped zero-coupon real curve is interpolated using the \citet{svensson1994estimating} parametrization, an extension of the \citet{nelson1987parsimonious} family adopted by ANBIMA and standard in the Brazilian term-structure literature \citep{varga2009teste}; we fit it by nonlinear least squares to match the observed ANBIMA indicative rates to within approximately 20~bp of yield error on typical days. The resulting real forward curves are evaluated at the same pillar maturities as the nominal curve.

\begin{remark}[Maturity mismatch and short-end filter]
\label{rem:maturity_mismatch}
The NTN-B grid is substantially coarser than the DI strip (9 to 14 outstanding maturities versus roughly 30 DI1 contracts). Below about six months, NTN-B is genuinely absent from primary issuance, so the Svensson fit at $\tau < 1$~year is pure extrapolation. For the real-curve PCA (Section~\ref{app:calibration}) we therefore restrict to pillars $\tau \in \{1, 2, 3, 5, 7, 10\}$ years; the full eight-pillar grid is retained only for the cross-sectional curve diagnostics. Decay parameters are shared with the nominal block defined in \Cref{subsec:shared_factor}.
\end{remark}

\subsubsection{Credit spread curves: JGP Idex-CDI and Idex-Infra}
\label{subsec:data_jgp}

The corporate credit term structures are obtained from the JGP Idex index families \citep{jgp_idex_cdi_methodology,jgp_idex_infra_methodology}. The \emph{Idex-CDI Core} tracks the spread of liquid, rated CDI-linked (CDI$^+$) corporate debentures over the DI curve, across a broad issuer universe. The \emph{Idex-Infra Core} tracks the spread of IPCA-linked (IPCA$^+$) infrastructure debentures over the NTN-B curve; these instruments are typically issued under \emph{Lei}~12.431, which grants an income-tax exemption to Brazilian individual investors holding infrastructure debentures \citep{lei_12431}, creating a distinct retail-driven demand clientele.

At each date we form a duration-bucketed spread curve at vertices $\tau_i \in \{1, 2, 3, 5\}$ years:
\begin{equation}\label{eq:spread_bucket}
  s_{mkt}(t, \tau_i) \;=\; \frac{\sum_{d \in B_i(t)} w_d(t)\, s_d(t)}{\sum_{d \in B_i(t)} w_d(t)},
\end{equation}
where $B_i(t)$ is the set of index constituents with modified duration within a half-width of vertex $\tau_i$ (half-widths $\pm 0.5$~y for the 1, 2, 3-year vertices and $\pm 1.0$~y for the 5-year vertex), $w_d(t)$ is the constituent weight, and $s_d(t)$ its daily quoted spread. The bucketing provides a stable term structure despite the sparse per-debenture quote history.

\paragraph{From observable quotes to instantaneous forwards.}
The theoretical objects in the three-currency framework are instantaneous credit-forward spreads $s^{CDI}(t,T)$ and $s^{IPCA}(t,T)$. The empirical objects in~\eqref{eq:spread_bucket} are duration-bucketed averages of constituent par-yield spreads. We treat $s_{mkt}(t,\tau_i)$ as a discrete-tenor proxy for $s^{CDI}(t,t+\tau_i)$ and $s^{IPCA}(t,t+\tau_i)$, with the understanding that par-yield-to-forward conversion at the constituent level would introduce a duration-dependent convexity correction of at most a few basis points across the maturity range used in the test. The IPCA-linked constituents carry a publication lag relative to nominal data; we align them to a common observation date using the publication-date convention. Robustness checks in Section~\ref{sec:robustness} vary the duration half-widths by $\pm 0.25$~y and shift the IPCA lag by $\pm 1$ business day; the mean within-issuer wedge $\bar\Delta_i$ moves by less than 5~bp under either perturbation, which bounds the contribution of this quote-to-forward conversion to the headline result.

\begin{remark}[Constituent universe]
\label{rem:constituents}
Over our sample, Idex-CDI Core contains 643 unique debentures from 181 issuers spanning sectors including \emph{Energia}, \emph{Financeiro}, \emph{Distribuidores}, \emph{Industrial}, \emph{Imobiliário}, \emph{Consumo}, \emph{Educação}, and \emph{Insumos}; per-date constituent counts average 195 (min 48, max 346). Idex-Infra Core contains 337 debentures from 108 issuers concentrated in \emph{Energia}, \emph{Saneamento}, \emph{Logística}, \emph{Telecom}, \emph{Bioenergia}, \emph{Papel e Celulose}, \emph{Produção Rural}, and \emph{Metais e Mineração}; per-date counts average 115 (min 38, max 301). The two universes partially overlap: 68 issuers have debentures in \emph{both} families, which plays a central role in the cross-model test of Section~\ref{subsec:dual_listed}.
\end{remark}

The 1-year IPCA$^+$ vertex is sparsely populated because Idex-Infra Core concentrates on longer-duration infrastructure debt: even under the wider half-window, only about 8\% of sample dates have five or more constituents in the 1-year bucket. Vertex $\tau = 2$~y has 42\% coverage, $\tau = 3$~y reaches 75\%, and $\tau = 5$~y is 97\% populated. We therefore use the 3-year pillar as the short-end reference for IPCA$^+$ realized-vol identification.

\subsubsection{Data alignment and sample overview}
\label{subsec:sample_overview}

We align the four time series (DI, NTN-B, Idex-CDI, Idex-Infra) on common business dates, excluding dates with missing observations in any of the rate series (the Idex-Infra spread curves have structural sparsity at short pillars, handled separately). The effective sample comprises 1{,}270 daily business-day observations from 2021-01-04 to 2026-02-05. PCA and correlation estimation use weekly (Friday-close) changes, yielding 155 weekly observations in the estimation window (2021-01-04 to 2023-12-29) and 110 in the out-of-sample window (2024-01-02 to 2026-02-05).

\Cref{tab:data_summary} reports summary statistics at the 3-year maturity, chosen because it is both densely populated across all four series and the maturity at which the dual-listed issuer test (Section~\ref{subsec:dual_listed}) is most informative.

\begin{table}[ht]
\centering
\caption{Summary statistics over 2021-01-04 to 2026-02-05 at the 3-year maturity. Rates in \% p.a.\ (continuous compounding); spreads in bp.}
\label{tab:data_summary}
\begin{tabular}{@{}lcccc@{}}
\toprule
\textbf{Variable} & \textbf{Mean} & \textbf{Std} & \textbf{Min} & \textbf{Max} \\
\midrule
$f^N(t,3)$ & 11.13 & 1.43 & 7.14 & 14.33 \\
$f^R(t,3)$ & 5.79 & 1.13 & 3.05 & 8.24 \\
Breakeven $f^N(t,3) - f^R(t,3)$ & 5.34 & 0.66 & 2.82 & 7.53 \\
$s^{CDI}_{mkt}(t,3)$ (bp) & 182.4 & 37.7 & 125.9 & 310.9 \\
$s^{IPCA}_{mkt}(t,3)$ (bp) & 24.2 & 38.7 & $-57.8$ & 123.6 \\
\bottomrule
\end{tabular}
\end{table}

\begin{figure}[ht]
\centering
\includegraphics[width=\linewidth]{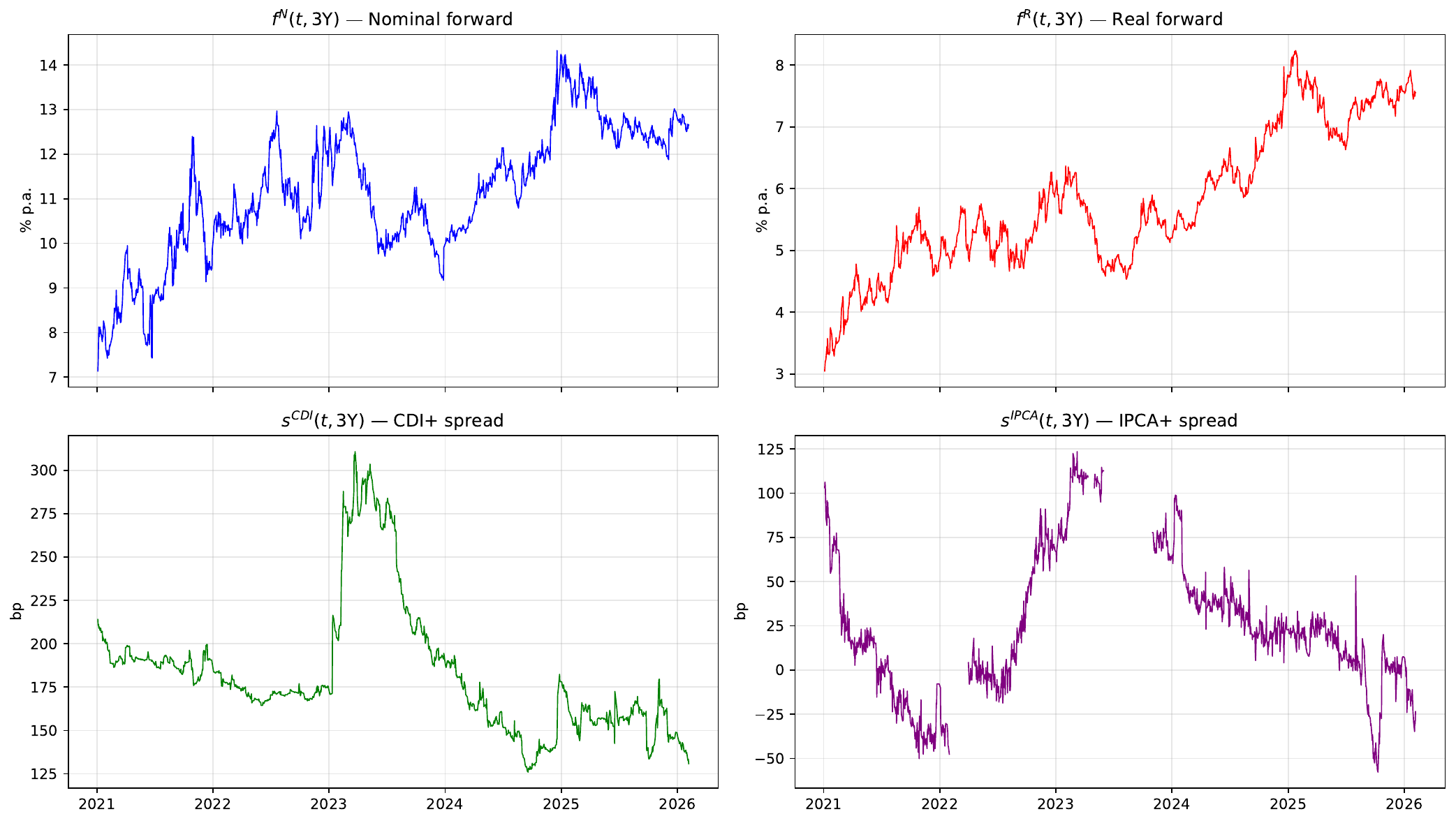}
\caption{Time series of the four 3-year series: DI futures yield, NTN-B real yield, $s^{CDI}_{mkt}(t,3)$, and $s^{IPCA}_{mkt}(t,3)$, over the full sample 2021-01-04 to 2026-02-05. The vertical dashed line separates the estimation window (through 2023-12-29) from the out-of-sample window.}
\label{fig:time_series}
\end{figure}

Figure~\ref{fig:time_series} (produced by the calibration code) plots the time series of the same four variables. The 2022 tightening cycle, the mid-2023 disinflation turn, and a localized credit-fund redemption episode in mid-2024 are all visible in the CDI$^+$ series. The IPCA$^+$ spread series is strikingly lower in level than the CDI$^+$ series, averaging about 24~bp over the NTN-B curve versus 182~bp for CDI$^+$ over DI. This 160~bp level difference foreshadows the tax-incentive finding of Section~\ref{subsec:dual_listed}.

\section{Within-issuer test of the triangle restriction}\label{sec:test}

\subsection{Calibration summary}\label{subsec:calibration}

The framework is calibrated as two parallel credit economies. Model~A treats CDI$^+$ spreads as the credit block; Model~B treats IPCA$^+$ spreads as the credit block. The nominal and real rate blocks and the inflation FX volatility $\sigma_I$ are shared across the two models; the spread loadings, rate-spread correlations, and the credit FX volatility $\sigma_J$ are model-specific.

Weekly curve changes identify $m_N=3$ nominal factors, $m_R=2$ real factors, and $m_S=2$ spread factors. The retained factors explain 84.3\% of nominal-rate variance, 90.5\% of real-rate variance, 82.4\% of CDI$^+$ spread-change variance, and 95.9\% of IPCA$^+$ spread-change variance (Model~B is fitted on pillars $\{2,3,5\}$~yr, with the sparsely populated 1-year vertex dropped). The inflation FX volatility, read from monthly IPCA log-returns after deseasonalization, is $|\sigma_I|=135$~bp/$\sqrt{\text{yr}}$. The credit FX volatilities $|\sigma_J^A|=80.7$ and $|\sigma_J^B|=75.9$~bp/$\sqrt{\text{yr}}$ are set by convention to the shortest-pillar realized spread volatility, not separately identified from market data; the within-issuer wedge analyzed in Section~\ref{subsec:dual_listed} does not depend on this choice (see Remark~\ref{rem:sigJ_convention} in Appendix~\ref{app:calibration}).

Cross-block dependence is read from raw factor-score correlations: within-block PCA orthogonalizes each block's loadings, so no further filtering is needed before computing the off-block entries of the $7\times 7$ correlation matrix $\hat\rho$. The full matrix is reported in Table~\ref{tab:corr_cdi}; an in-sample/out-of-sample comparison for the four pairs of economic interest appears in Table~\ref{tab:oos_corr}.

Out-of-sample coverage is adequate for the objects used in the main test. One-week 90\% intervals cover nominal-rate changes at 90.9 to 98.2\% across reported maturities and CDI$^+$ spread changes at 85.5 to 92.7\%. Real-rate coverage falls to 69.1 to 75.5\% at the 5- and 7-year maturities because the 2024 to 2026 long-end NTN-B moves were larger than in the estimation window. Spread innovations are heavier-tailed than rates, which is the main distributional limitation of the Gaussian HJM specification.

\subsection{Dual-listed test}\label{subsec:dual_listed}

For an issuer $i$ observed in both the Idex-CDI and Idex-Infra constituent lists, the common-$f^C$ hypothesis under which Models A and B agree predicts that
\begin{equation}\label{eq:dual_listed}
  \Delta_i(t) \;\equiv\; \bar s^{IPCA}_i(t) - \bar s^{CDI}_i(t) - [f^N(t, \bar \tau_i) - f^R(t, \bar \tau_i)] \;\approx\; 0,
\end{equation}
up to maturity-matching and idiosyncratic noise, where bars denote per-date weight averages across the issuer's tickers in each family and $\bar\tau_i$ is the average duration matched to the nearest vertex of the breakeven curve.

\begin{table}[ht]
\centering
\caption{Dual-listed issuer test: per-issuer time-series statistics of $\Delta_i$ (bp) across the 15 dual-listed issuers with the longest joint observation windows. Positive $\Delta_i$ means IPCA$^+$ spread is \emph{higher} than the triangle-consistent level; negative means lower.}
\label{tab:dual_listed}
\small
\begin{tabular}{@{}lrrrr@{}}
\toprule
Issuer & Mean & Median & Std & $n_{\text{obs}}$ \\
\midrule
Eletrobras                                          & $-629.6$ & $-632.6$ & 61.8 & 1{,}270 \\
Eneva                                               & $-660.2$ & $-659.0$ & 53.8 & 1{,}270 \\
Energisa Mato Grosso - Distribuidora de Energia     & $-625.5$ & $-629.4$ & 72.8 & 1{,}210 \\
Cemig Distribuicao                                  & $-604.8$ & $-601.9$ & 74.7 & 1{,}207 \\
CTEEP                                               & $-663.8$ & $-668.6$ & 63.0 & 1{,}190 \\
Sabesp                                              & $-647.2$ & $-641.7$ & 68.3 & 1{,}190 \\
Energisa                                            & $-637.8$ & $-637.6$ & 68.7 & 1{,}149 \\
BRK AMBIENTAL                                       & $-693.7$ & $-692.4$ & 49.5 & 1{,}141 \\
CCR                                                 & $-643.5$ & $-635.4$ & 73.1 & 1{,}022 \\
Algar Telecom                                       & $-625.2$ & $-627.7$ & 65.2 & 1{,}003 \\
Coelce                                              & $-607.5$ & $-606.8$ & 59.3 &    961 \\
Taesa                                               & $-663.1$ & $-654.2$ & 88.0 &    943 \\
Coelba                                              & $-665.1$ & $-667.1$ & 47.8 &    943 \\
Copel Geracao e Transmissao                         & $-642.1$ & $-648.2$ & 70.8 &    898 \\
AES Tiete Energia                                   & $-596.8$ & $-594.3$ & 61.8 &    815 \\
\midrule
\textbf{Summary across 15 issuers} & & & & \\
Mean of per-issuer means            & \multicolumn{4}{l}{$-640.4$~bp} \\
Across-issuer std of means          & \multicolumn{4}{l}{$26.4$~bp} \\
Mean of per-issuer std              & \multicolumn{4}{l}{$65.2$~bp} \\
Total observations                  & \multicolumn{4}{l}{16{,}212} \\
\bottomrule
\end{tabular}
\end{table}

\begin{figure}[ht]
\centering
\includegraphics[width=\linewidth]{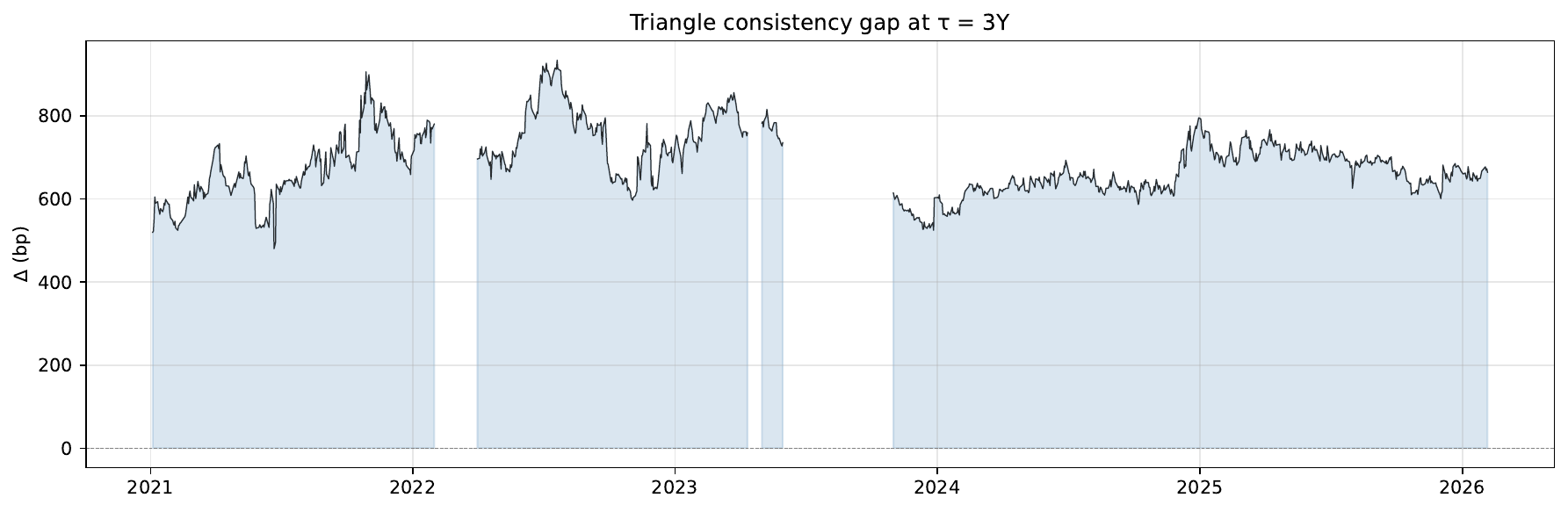}
\caption{Time series of the within-issuer triangle residual $\Delta_i(t)$ for the 15 dual-listed Brazilian issuers between 2021-01-04 and 2026-02-05. Each line corresponds to one issuer's $\Delta_i$ at the 3-year tenor. The residual is uniformly negative and clusters in a narrow band around $-640$~bp, with cross-sectional standard deviation of issuer means equal to 26~bp, well below the within-issuer time-series standard deviation of approximately 65~bp. The pattern is stable through the 2021--2023 BCB tightening cycle and the 2024--2026 easing phase.}
\label{fig:triangle_gap}
\end{figure}

Two features dominate Table~\ref{tab:dual_listed} and Figure~\ref{fig:triangle_gap}. Every dual-listed issuer has $\Delta_i$ negative and clustered in the $-600$ to $-700$~bp range: the IPCA$^+$ spread on each issuer's infrastructure debentures is systematically below the triangle-consistent level. The cross-sectional dispersion of issuer means is 26~bp, substantially smaller than the within-issuer time-series standard deviation of $\Delta_i$, which averages 65~bp. The wedge is a property of the market segment, not the issuer.

The pattern rules out the common-$f^C$ hypothesis even within the subset of issuers that trade in both segments. Sign, magnitude, and cross-sectional uniformity all point to the retail-tax channel of \emph{Lei}~12.431: Brazilian individual investors holding qualifying infrastructure debentures with maturity above 720 days pay no income tax on coupons and amortization, whereas the same investors would face a 15\% withholding tax on income from a comparable long-tenor non-exempt CDI$^+$ debenture \citep{lei_12431,lei_11033}. Section~\ref{subsec:decomposition} quantifies the tax channel against the observed wedge under both a breakeven-adjusted reading and a raw-spread reading, and treats the residual as a candidate for the institutional, contractual, and liquidity effects discussed there.

In modelling terms, the framework's Models A and B must be treated as distinct credit economies. Forcing both segments into a single $f^C$ would absorb a 640~bp deterministic wedge that is structural rather than residual measurement error.

\subsection{Post-tax indifference benchmark}\label{subsec:tax_parity}

The triangle-test residual $\Delta_i$ defined in \eqref{eq:Delta_def} measures how far the within-issuer comparison departs from the common-credit-economy null. To quantify how much of any such departure is attributable to a known regulatory channel, we derive the prediction of post-tax indifference for a Brazilian retail investor under the \emph{Lei} 11.033 progressive schedule (15\% on debentures held longer than 720 days) with the \emph{Lei} 12.431 exemption applied to qualifying infrastructure debentures.

\paragraph{Indifference condition.}
Let $y^{IPCA}(t,\tau) = f^N(t,\tau) + s^{IPCA}(t,\tau)$ denote the nominal yield on an IPCA-linked corporate debenture, having passed inflation through the breakeven, and let $y^{CDI}(t,\tau) = f^N(t,\tau) + s^{CDI}(t,\tau)$ denote the nominal yield on a CDI-linked corporate debenture. A retail investor exempt on the IPCA leg under \emph{Lei}~12.431 and taxed at $\tau_{PF} = 0.15$ on the CDI leg is indifferent at the maturity-matched margin when after-tax nominal yields agree:
\begin{equation}\label{eq:tax_indifference}
  y^{IPCA}(t,\tau) = (1-\tau_{PF})\,y^{CDI}(t,\tau).
\end{equation}

\paragraph{Prediction for the raw spread difference (Candidate A).}
Expanding~\eqref{eq:tax_indifference} and rearranging,
\begin{equation}\label{eq:tau_fiscal_raw}
  s^{IPCA}(t,\tau) - s^{CDI}(t,\tau) \;=\; -\tau_{PF}\,\bigl[f^N(t,\tau) + s^{CDI}(t,\tau)\bigr].
\end{equation}
The post-tax indifference prediction for the raw spread difference is the tax wedge $\tau_{PF}$ applied to the gross CDI corporate yield. At the 3-year sample means ($f^N = 11.13\%$, $s^{CDI} = 1.82\%$) this evaluates to $-0.15 \times 12.95\% = -194$~bp.

\paragraph{Prediction for the breakeven-adjusted wedge (Candidate B).}
Substituting~\eqref{eq:tau_fiscal_raw} into the definition $\Delta(t,\tau) = s^{IPCA}(t,\tau) - s^{CDI}(t,\tau) - [f^N(t,\tau) - f^R(t,\tau)]$ gives
\begin{equation}\label{eq:tau_fiscal_delta}
  \tau_{fiscal}(t,\tau) \;=\; -\bigl[f^N(t,\tau) - f^R(t,\tau)\bigr] \;-\; \tau_{PF}\bigl[f^N(t,\tau) + s^{CDI}(t,\tau)\bigr].
\end{equation}
This is the linear post-tax benchmark for the within-issuer wedge $\Delta$. The first term is the deterministic breakeven removal that the definition of $\Delta$ requires; the second is the behavioral tax-wedge component. At the 3-year sample means ($f^N - f^R = 5.34\%$), $\tau_{fiscal} = -534 - 194 = -728$~bp at the aggregate level and $-733$~bp averaged across the 15 dual-listed issuers' durations.

Equations~\eqref{eq:tau_fiscal_raw} and~\eqref{eq:tau_fiscal_delta} are arithmetically equivalent: each predicts the same residual $\eta$ relative to its corresponding observable. Section~\ref{subsec:decomposition} reports both readings of the per-issuer decomposition.

\paragraph{Finite-horizon correction.}
Equation~\eqref{eq:tax_indifference} is an instantaneous-yield condition. The exact retail-investor problem applies the tax to the realized payoff over the holding period $D$, replacing $(1-\tau_{PF})\,y^{CDI}$ in~\eqref{eq:tax_indifference} with
\begin{equation}\label{eq:tax_parity_exact}
  y^{after}_{CDI}(y, D) = \frac{1}{D}\,\ln\!\left[1 + (1-\tau_{PF})\bigl(e^{y D} - 1\bigr)\right].
\end{equation}
At the per-issuer means in our sample the correction lies in the 25--45~bp range, averaging $29$~bp across the 15 issuers. The exact-form analog of~\eqref{eq:tau_fiscal_delta} shifts the aggregate $\tau_{fiscal}$ prediction from $-728$ to approximately $-699$~bp, with corresponding shifts at each issuer's duration.

\paragraph{Scope.}
The benchmark assumes the marginal IPCA$^+$ investor is retail and tax-exempt under \emph{Lei}~12.431. It does not price institutional segmentation on either leg, contractual differences between general-purpose CDI$^+$ debentures and use-of-proceeds-restricted infrastructure debentures, or residual liquidity premia. Section~\ref{subsec:decomposition} treats those as candidates for the residual $\eta$.

\subsection{Decomposition of the within-issuer wedge}\label{subsec:decomposition}

Section~\ref{subsec:tax_parity} derived the post-tax indifference benchmark $\tau_{fiscal}$ for the within-issuer wedge under~\eqref{eq:tau_fiscal_delta} (linear) and its finite-horizon counterpart \eqref{eq:tax_parity_exact}. Applied at each issuer's average duration with the issuer's time-averaged CDI$^+$ spread, the linear form averages $\overline{\tau_{fiscal}} = -733.2$~bp across the 15 dual-listed names, against $\bar\Delta_i = -640.4$~bp; the finite-horizon form averages $-698.9$~bp. The corresponding mean residuals $\eta_i = \Delta_i - \tau_{fiscal}(i)$ are $+92.8$~bp (linear, across-issuer s.e.\ 4.9~bp) and $+58.5$~bp (exact, s.e.\ 4.9~bp). The exact correction narrows but does not eliminate the residual.

\Cref{tab:wedge_decomp} reports the decomposition $\Delta_i = \tau_{fiscal}(i) + \eta_i$ at the issuer level.

\begin{table}[ht]
\centering
\begin{tabular}{l rrrr}
\toprule
Issuer & $\Delta_i$ & $\tau_{fiscal}(i)$ & Residual & Avg.\ dur.\ (yr) \\
 & (bp) & (bp) & (bp) & \\
\midrule
Eletrobras & -630 & -731 & +101 & 3.7 \\
Eneva & -660 & -749 & +89 & 4.6 \\
Energisa Mato Grosso - Distri... & -626 & -724 & +98 & 3.4 \\
Cemig Distribuicao & -605 & -718 & +114 & 3.1 \\
CTEEP & -664 & -741 & +77 & 4.4 \\
Sabesp & -647 & -732 & +84 & 3.8 \\
Energisa & -638 & -730 & +92 & 3.8 \\
BRK AMBIENTAL - PARTICIPACOES... & -694 & -773 & +79 & 4.9 \\
CCR & -644 & -734 & +90 & 3.7 \\
Algar Telecom & -625 & -737 & +112 & 3.9 \\
Coelce & -608 & -696 & +89 & 2.4 \\
Taesa - Transmissora Alianca ... & -663 & -738 & +75 & 4.3 \\
Coelba & -665 & -741 & +76 & 4.2 \\
Copel Geracao e Transmissao & -642 & -735 & +93 & 4.1 \\
AES Tiete Energia & -597 & -719 & +123 & 3.2 \\
\midrule
\textbf{Mean} & \textbf{-640.4} & \textbf{-733.2} & \textbf{+92.8} & 3.8 \\
Standard error & (6.8) & (4.3) & (3.7) & \\
\bottomrule
\end{tabular}

\caption{Per-issuer decomposition of the within-issuer wedge into a pure retail-tax component and a residual. $\tau_{fiscal}$ computed from \eqref{eq:tau_fiscal_delta} at each issuer's average duration; $\eta = \Delta_i - \tau_{fiscal}$. All values in basis points. Standard errors on the bottom row are across-issuer standard errors of the mean.}
\label{tab:wedge_decomp}
\end{table}

The mean tax-parity prediction across the 15 issuers is $-733$~bp; the mean observed wedge is $-640$~bp; the mean residual $\eta$ is $+93$~bp, with cross-sectional standard error of $3.7$~bp. In $\Delta$ basis the retail-tax channel accounts for the sign of the wedge and most of its level, with a small positive residual.

\paragraph{Raw-spread companion.}
The $\Delta$-basis prediction $-733$~bp decomposes mechanically into a deterministic breakeven removal $-(f^N - f^R) \approx -534$~bp at the sample 3-year mean and a behavioral tax-wedge component $-\tau_{PF}(f^N + s^{CDI}) \approx -194$~bp. Only the second component is a tax mechanism; the first is the unit conversion that defines $\Delta$. In raw-spread units, the indifference condition predicts $s^{IPCA}_i - s^{CDI}_i \approx -194$~bp at the sample 3-year means, against an observed aggregate of $-158$~bp and a within-issuer mean of approximately $-106$~bp at the 15 dual-listed names' average duration. The tax mechanism therefore over-predicts the behavioral spread gap by 50 to 80\%, with the residual carrying the opposite sign of the $\Delta$-basis residual. The two readings are arithmetically identical (the wedge $\Delta$ is the raw spread difference minus a deterministic breakeven), but they differ in what the reader sees as the size of the tax channel relative to the size of the unexplained part: the retail-tax channel is the dominant mechanism behind the wedge, but it is not the sole one, and the size of the residual under the behavioral reading is consistent with non-trivial contributions from institutional pricing and segment-specific structural and liquidity factors discussed below.

The positive $\Delta$-basis residual fits three non-exclusive channels we cannot separate cleanly at $N = 15$. First, institutional investors are partly marginal on the CDI$^+$ side. Pension portfolios, credit funds, and tax-privileged investment vehicles hold a non-trivial share of CDI$^+$ debenture flow, and their pricing compresses the pre-tax yield differential that pure retail parity would require. Second, IPCA$^+$ \emph{Lei}~12.431 infrastructure debentures carry contractual features that do not generally appear in same-issuer general-purpose CDI$^+$ paper: use-of-proceeds restrictions, project ring-fencing, and covenant packages designed around long-tenor infrastructure cash flows. Sharing a legal entity does not guarantee identical claim structures, and these structural differences absorb part of the residual that a pure tax-clientele model would otherwise have to attribute to marginal-investor segmentation. Third, IPCA$^+$ infrastructure debentures have lower secondary-market turnover than the CDI$^+$ core pool, so a small illiquidity discount on the IPCA$^+$ side pushes $s^{IPCA}_i$ above the post-tax-indifference level.

The cross-sectional regression
\begin{equation}
  \Delta_i = \alpha + \beta\, \tau_{fiscal}(i) + \gamma\, (\text{dur}_i^{infra} - \text{dur}_i^{cdi}) + \varepsilon_i
\end{equation}
appears in \Cref{tab:wedge_regression}. The univariate $\beta$ is $1.39$ with standard error $0.21$. Adding the duration differential gives $\beta = 1.30$ with standard error $0.30$ and a duration coefficient that is not significant at any conventional level. $H_0: \beta = 1$ is not rejected at the 5\% level under the two-regressor specification ($p = 0.34$). $R^2$ is $0.76$ or higher in both specifications. The retail-tax channel tracks the cross-sectional dispersion of $\Delta_i$, not just the mean. With $N = 15$ and $\tau_{fiscal}(i)$ ranging across only about 50~bp (the breakeven term is shared across issuers), the regression is a sensitivity check rather than a primary identification; the mean decomposition is the headline.

\begin{table}[ht]
\centering
\begin{tabular}{l rr}
\toprule
 & M1 & M2 \\
\midrule
Intercept $\alpha$ & $374.8$ & $323.5$ \\
  & $(157.1)$ & $(206.6)$ \\
$\tau_{fiscal}$ ($\beta$) & $\mathbf{1.385}$ & $\mathbf{1.302}$ \\
  & $(0.214)$ & $(0.302)$ \\
Duration differential $(\gamma)$ & --- & $-3.26$ \\
  &  & $(8.10)$ \\
\midrule
$R^2$         & $0.763$ & $0.766$ \\
Adj.\ $R^2$   & $0.744$ & $0.727$ \\
$p$-value: $H_0:\beta=1$ & $0.096$ & $0.337$ \\
$N$            & $15$ & $15$ \\
\bottomrule
\end{tabular}

\caption{Cross-sectional regression of the within-issuer wedge on the theoretical retail-tax prediction, optionally controlling for the infra-minus-CDI duration differential per issuer. Standard errors in parentheses. The hypothesis $H_0: \beta = 1$ is the test that the retail-tax mechanism fully explains the cross-sectional variation in $\Delta_i$, against the alternative that either over- or under-sensitivity reflects additional channels.}
\label{tab:wedge_regression}
\end{table}

Two caveats. The 15\% rate applies to debentures held longer than 720 days. All securities in our sample sit comfortably above that threshold, so the rate is uniform across issuers. Were we to extend the test to short-tenor primary issuance, the rate schedule would matter. The decomposition also treats retail as marginal on both sides. If institutional investors set CDI$^+$ prices at the margin, $\tau_{fiscal}$ overstates the pure tax channel and part of the residual $\eta_i$ should be reassigned to that channel rather than to a clientele or liquidity premium. Sharper identification would require ownership data from CVM/ANBIMA holdings surveys to estimate the marginal-investor share of each market segment.


\section{Robustness}\label{sec:robustness}

This section reports the two robustness checks that bear directly on the economic interpretation of the within-issuer wedge: stability of the wedge across the 2021--2023 BCB tightening and 2024--2026 easing regimes, and sensitivity of the wedge to duration-matching and bucketing conventions.

\subsection{Sub-sample split: tightening versus easing regimes}\label{subsec:robust_covid}

The 2021-01 onward sample does not include the March to December 2020 COVID episode (Idex-CDI data begins 2021-01-04). The economically meaningful split is between the 2021-01 to 2023-12 BCB tightening cycle, during which the Selic rate rose from 2\% to 13.75\%, and the 2024-01 to 2026-02 easing phase. \Cref{tab:subsample_wedge} recomputes the within-issuer wedge $\Delta_i$ on each sub-sample separately, holding the universe of 15 dual-listed issuers fixed.

\begin{table}[ht]
\centering
\begin{tabular}{lrrrrr}
\toprule
Regime & Mean $\Delta_i$ (bp) & SD across $i$ (bp) & SE (bp) & $N$ issuers & Avg obs/issuer \\
\midrule
Tightening (2021-01 to 2023-12) & -651.2 & 37.6 & 9.7 & 15 & 573 \\
Easing (2024-01 to 2026-02) & -629.4 & 26.2 & 6.8 & 15 & 508 \\
Full & -640.4 & 26.4 & 6.8 & 15 & 1081 \\
\bottomrule
\end{tabular}

\caption{Dual-listed wedge $\Delta_i = \bar s^{IPCA}_i - \bar s^{CDI}_i - (f^N - f^R)$ by monetary-policy regime. The same 15 dual-listed issuers appear in all three rows, with date coverage restricted to the relevant window.}
\label{tab:subsample_wedge}
\end{table}

The mean wedge is $-651$~bp in the tightening regime and $-629$~bp in the easing regime, a difference of 22~bp with combined standard error of about 12~bp. The regime-to-regime difference is marginally distinguishable from zero ($|z| \approx 1.85$) but small relative to the wedge level itself. The cross-issuer standard deviation of $\Delta_i$ is larger in the tightening phase (38~bp) than in the easing phase (26~bp), reflecting greater idiosyncratic spread volatility during the hiking cycle.

The retail-tax-parity prediction of \Cref{subsec:decomposition}, recomputed on each sub-sample, gives $-725$~bp for the tightening regime and $-745$~bp for the easing regime. The pure-tax prediction therefore moves in the direction opposite to the observed wedge: higher nominal yields in the easing phase imply a larger tax benefit and therefore a deeper predicted wedge, while the observed wedge is slightly shallower in easing than in tightening. The residual $\eta = \Delta - \tau_{fiscal}$ is larger in the easing phase ($+116$~bp) than in the tightening phase ($+74$~bp). A natural reading is that retail investors became less dominant on the CDI$^+$ side as credit-fund inflows picked up through 2024 to 2025, or that the illiquidity discount on infrastructure debentures widened around the mid-2024 stress episode.

\subsection{Duration matching and bucketing}\label{subsec:robust_matching}

The headline $\Delta_i = -640$~bp is computed at the issuer level from constituent-level spread data, matching each issuer's average duration to the nearest of the four breakeven pillars $\tau \in \{1, 2, 3, 5\}$ years. Two sensitivities matter: whether the matching rule itself biases the wedge, and whether the duration half-widths underlying the aggregate-index curves contaminate the within-issuer test.

\Cref{tab:duration_matching} reruns the dual-listed test under two alternative matching schemes: linear interpolation of the breakeven between pillars at the issuer's mid-duration, and separate interpolation at each side's own duration followed by averaging of the two breakevens. The mean $\Delta_i$ changes by less than 5~bp under either alternative, and the average across-scheme range within a single issuer is 11~bp. The within-issuer wedge is not an artefact of the matching rule.

\begin{table}[ht]
\centering
\begin{tabular}{lrr}
\toprule
Matching scheme & Mean $\Delta_i$ (bp) & SE (bp) \\
\midrule
Nearest pillar at mid-duration (baseline) & -640.4 & 6.8 \\
Linear interpolation of BE at mid-duration & -641.6 & 6.3 \\
Split-side interpolation (avg of two BEs) & -635.9 & 4.9 \\
\midrule
\multicolumn{3}{l}{Mean across-scheme range per issuer: 10.9~bp} \\
\bottomrule
\end{tabular}

\caption{Duration-matching sensitivity of the dual-listed wedge. All schemes operate on constituent-level data and differ only in how the breakeven forward is matched to each issuer's average duration.}
\label{tab:duration_matching}
\end{table}

The duration half-widths that define the aggregate CDI$^+$ and IPCA$^+$ index curves at pillar maturities (baseline $\pm 0.5$~yr for $\tau \in \{1, 2, 3\}$ and $\pm 1.0$~yr for $\tau = 5$) play no role in the within-issuer test, which uses constituent-level spread and duration data directly. They do affect the index-level aggregate gap $\Delta_{agg}$. Recomputed at $\tau = 3$~years under tight ($\pm 0.25$~yr) and wide ($\pm 0.75$~yr) half-width specifications, the mean aggregate gap varies by less than 2~bp from the baseline.

\section{Discussion}\label{sec:discussion}

The data reject the shared-credit-economy null at every dual-listed issuer. A within-issuer wedge of $-640$~bp at the 3-year tenor, with cross-sectional standard deviation of 26~bp across 15 issuers, is too large for measurement noise and too uniform across issuers to be a name-specific credit effect. The mechanism is tax segmentation under \emph{Lei}~12.431: individual investors face zero income tax on coupons and amortisation of qualifying IPCA-linked infrastructure debentures while bearing the 15\% withholding rate of \emph{Lei}~11.033 on comparable CDI-linked corporate paper. The exemption sustains a retail clientele whose reservation yield on the IPCA-linked side sits permanently below the institutional reservation yield on the CDI-linked side.

The Brazilian wedge sits in the same structural family as the U.S.\ tax-exempt municipal market. Both rest on a tax-induced retail clientele that institutional investors cannot arbitrage because the exemption is investor-specific. Identification here is sharper than in the muni-versus-taxable comparison for two reasons. The wedge is computed within issuer rather than across heterogeneous universes, so rating, sector, and duration composition drop out. And the triangle identity removes nominal-real rate-regime variation before the comparison, so the residual is a clean clientele-segmentation object. The retail-tax benchmark reproduces the sign and bulk of the wedge under the breakeven-adjusted reading (predicted $-733$~bp against observed $-640$~bp, residual $+93$~bp), and over-predicts the raw spread differential under the behavioural reading. The two readings of the decomposition are arithmetically identical; together they bracket the tax channel as the dominant mechanism while leaving room for a residual that institutional pricing on the CDI side, contractual asymmetries, and liquidity differences can plausibly explain.

The Gaussian HJM specification is adequate for the rate blocks and useful for credit spreads, but spread innovations are heavier-tailed than the specification supports, and the mid-2024 credit-fund redemption episode visibly stresses the spread block. The main empirical limitation is therefore distributional, not the no-arbitrage restriction: a stochastic-volatility overlay or a jump component on the spread block would better fit stress episodes without changing the three-currency triangle.

\paragraph{Relation to prior Brazilian estimates.}
\citet{pereira2021nota} estimate the fiscal pass-through of \emph{Lei}~12.431 from the issuer's side, measuring the fraction of foregone federal tax revenue that transfers to project funding costs. \citet{brand2022debentures} tests no-arbitrage appropriation of the same fiscal benefit between incentivised and non-incentivised debentures, finding that the secondary-market spread conforms to the no-arbitrage prediction. Both studies condition on the issuer or instrument. The within-issuer triangle test reads the wedge from the investor's pricing side under the no-arbitrage restriction of the three-currency framework: 640~bp at the 3-year tenor with 26~bp cross-sectional standard deviation across 15 issuers. The investor-side estimate and the issuer-side fiscal-efficiency estimate are complementary lenses on the same regulatory channel.

\paragraph{Three limitations.}
First, the marginal-investor identity is inferred from the size of the residual rather than estimated structurally. Separating the institutional, contractual, and liquidity contributions to the residual would require either holdings data on individual debentures (CVM aggregate ownership disclosures, for example) or a structural model of clientele formation. Second, the test is confined to the 15 dual-listed issuers with at least 800 days of joint observation, and selection into both index families is non-random; the population wedge estimate is conditional on this selection. Third, the Gaussian HJM specification fits rates well but understates the tail behaviour of credit-spread innovations, as the mid-2024 credit-fund redemption episode illustrates. A stochastic-volatility overlay or a jump component on the spread block would address the distributional limitation without altering the HJM drift logic; a deeper Idex-Infra panel would sharpen Model~B factor identification.

\paragraph{Future work.}
The natural next step is to extend the dual-listed test as additional infrastructure debentures meet the 800-day joint-observation threshold, exploiting ownership disclosures to bound the marginal-investor mix on each segment directly. A parallel calibration to the Idex-CDI Low-Rated subuniverse would test whether the within-issuer wedge logic extends to a rating-quality dimension. A future regulatory revision to \emph{Lei}~12.431 would create a natural-experiment opportunity: the framework handles this by calibrating before and after the policy break as separate credit-economy parameter sets.

\paragraph{Conclusion.}
The evidence supports a two-credit-economy reading of Brazilian corporate debenture markets. CDI$^+$ and IPCA$^+$ spreads cannot be reconciled by a single credit forward curve once breakeven inflation is removed. The remaining within-issuer wedge is large, stable, and quantitatively aligned with the retail tax exemption on infrastructure debentures. The three-currency HJM construction provides both the no-arbitrage diagnostic that makes the within-issuer comparison precise and the calibration architecture that prices each segment in its own credit economy.


\appendix

\section{HJM dynamics and change-of-measure derivations}\label{app:hjm}

This appendix develops the formal no-arbitrage conditions, the change-of-measure structure connecting the three pricing measures, and the resulting forward-rate dynamics under the nominal measure. The end product is the forward-rate system that the calibration of Appendix~\ref{app:calibration} estimates and the simulation in Section~\ref{sec:simulation} integrates.

\subsection{Overview of the system under \texorpdfstring{$\QN$}{QN}}\label{subsec:hjm_overview}

Under the nominal pricing measure $\QN$, no-arbitrage between the three economies requires that the discounted values of the real and credit numeraires be martingales. With deterministic volatility loadings and the inner-product convention of \Cref{rem:inner}, the forward-rate system under $\QN$ takes the form
\begin{equation}\label{eq:system}
\begin{aligned}
df^N(t,T) &= \sigvec{N}(t,T)\cdot\Avec{N}(t,T)\,\dd t
          + \sigvec{N}(t,T)\cdot\dd W^N_t,\\
df^R(t,T) &= \big[\sigvec{R}(t,T)\cdot\Avec{R}(t,T)
          - \sigvec{R}(t,T)\cdot\sigI(t)\big]\dd t
          + \sigvec{R}(t,T)\cdot\dd W^N_t,\\
df^C(t,T) &= \big[\sigvec{C}(t,T)\cdot\Avec{C}(t,T)
          - \sigvec{C}(t,T)\cdot\sigJ(t)\big]\dd t
          + \sigvec{C}(t,T)\cdot\dd W^N_t,
\end{aligned}
\end{equation}
where $\Avec{j}(t,T)=\int_t^T \sigvec{j}(t,u)\dd u$. The real-curve drift correction is the Jarrow-Yildirim inflation term; the credit-curve correction is its credit-FX analogue. The rest of the appendix derives \eqref{eq:system} formally and provides the bond-price dynamics consistent with it.

\subsection{Martingale conditions and exchange-rate dynamics}
\label{sec:noarb_full}

\subsubsection{Martingale conditions}
\label{subsec:martingale}

We designate the nominal economy as ``domestic.'' Under $\QN$, absence of arbitrage requires that the discounted nominal values of foreign-economy numeraires be martingales.

\begin{proposition}[Martingale conditions]
\label{prop:martingale}
The following processes are martingales under the measures indicated:
\begin{align}
  \frac{I_t \, B^R_t}{B^N_t} &\quad \text{is a $\QN$-martingale}, \label{eq:mart_IR} \\[6pt]
  \frac{J_t \, B^C_t}{B^N_t} &\quad \text{is a $\QN$-martingale}, \label{eq:mart_JC} \\[6pt]
  \frac{K_t \, B^C_t}{B^R_t} &\quad \text{is a $\QR$-martingale}. \label{eq:mart_KC}
\end{align}
\end{proposition}

\begin{proof}
Conditions \eqref{eq:mart_IR} and \eqref{eq:mart_JC} are the standard two-currency no-arbitrage requirements applied to the nominal-real and nominal-credit pairs, respectively \citep[see][Chapter~14]{musiela2005martingale}. For \eqref{eq:mart_KC}: define the Radon-Nikodym derivative connecting $\QR$ to $\QN$ as
\[
  \left.\frac{d\QR}{d\QN}\right|_{\Ft} = \frac{I_t B^R_t / B^N_t}{I_0 B^R_0 / B^N_0}.
\]
Under $\QR$, the process $K_t B^C_t / B^R_t$ equals $(J_t B^C_t / B^N_t) \cdot (B^N_t / I_t B^R_t)$. Since $J_t B^C_t / B^N_t$ is a $\QN$-martingale and $I_t B^R_t / B^N_t$ defines the measure change, the standard change-of-measure formula for conditional expectations \citep[see][Sections~1.6 and 5.2]{shreve2004stochastic} implies that $K_t B^C_t / B^R_t$ is a $\QR$-martingale. Thus \eqref{eq:mart_KC} is not independent: it follows from \eqref{eq:mart_IR} to \eqref{eq:mart_JC} and the Girsanov change of measure. We state it explicitly to close the triangle.
\end{proof}

\subsubsection{Exchange-rate dynamics under \texorpdfstring{$\QN$}{QN}}
\label{subsec:fx_dyn}

We specify lognormal dynamics for the two primary exchange rates under $\QN$.

\begin{proposition}[FX dynamics]
\label{prop:fx_dynamics}
Under $\QN$, the martingale conditions \eqref{eq:mart_IR} to \eqref{eq:mart_JC} uniquely determine the drift coefficients:
\begin{align}
  \frac{dI_t}{I_t} &= \big(r^N_t - r^R_t\big)\dd t + \sigI(t)\cdot dW^N_t, \label{eq:dI} \\[6pt]
  \frac{dJ_t}{J_t} &= \big(r^N_t - r^C_t\big)\dd t + \sigJ(t)\cdot dW^N_t, \label{eq:dJ}
\end{align}
where $\sigI(t),\,\sigJ(t)\in\mathbb{R}^m$ are deterministic volatility vectors.
\end{proposition}

\begin{proof}
We prove \eqref{eq:dI}; the argument for \eqref{eq:dJ} is identical with $(R,I)$ replaced by $(C,J)$.

Write $dI_t/I_t = \mu_I(t)\dd t + \sigI(t)\cdot dW^N_t$ for some drift $\mu_I(t)$ to be determined. Consider the process $M_t = I_t B^R_t / B^N_t$. Since $dB^j_t = r^j_t B^j_t \dd t$, an application of It\^o's product rule gives:
\begin{align*}
  \frac{dM_t}{M_t}
  &= \frac{dI_t}{I_t} + \frac{dB^R_t}{B^R_t} - \frac{dB^N_t}{B^N_t} \\
  &= \big[\mu_I(t) + r^R_t - r^N_t\big]\dd t + \sigI(t)\cdot dW^N_t.
\end{align*}
For $M_t$ to be a $\QN$-martingale, the drift must vanish:
\[
  \mu_I(t) + r^R_t - r^N_t = 0 \quad\implies\quad \mu_I(t) = r^N_t - r^R_t. \qedhere
\]
\end{proof}

\begin{proposition}[Cross-rate dynamics with It\^o correction]
\label{prop:cross_rate}
Since $K_t = J_t / I_t$, It\^o's formula for the ratio of two semimartingales yields:
\begin{equation}\label{eq:dK}
  \frac{dK_t}{K_t} = \Big(r^R_t - r^C_t + \big|\sigI(t)\big|^2 - \sigI(t)\cdot\sigJ(t)\Big)\dd t + \big(\sigJ(t) - \sigI(t)\big)\cdot dW^N_t.
\end{equation}
\end{proposition}

\begin{proof}
Define $\varphi(x,y) = x/y$ so that $K_t = \varphi(J_t, I_t)$. By It\^o's formula applied to $\varphi$:
\begin{align*}
  d\varphi &= \frac{1}{I_t}\,dJ_t - \frac{J_t}{I_t^2}\,dI_t + \frac{J_t}{I_t^3}\,d\langle I\rangle_t - \frac{1}{I_t^2}\,d\langle J, I\rangle_t.
\end{align*}
Dividing through by $K_t = J_t/I_t$ and substituting:
\begin{itemize}[leftmargin=2em]
  \item $d\langle I\rangle_t = I_t^2 |\sigI(t)|^2\dd t$,
  \item $d\langle J, I\rangle_t = I_t J_t \,\sigI(t)\cdot\sigJ(t)\dd t$,
\end{itemize}
we obtain:
\begin{align*}
  \frac{dK_t}{K_t} &= \frac{dJ_t}{J_t} - \frac{dI_t}{I_t} + |\sigI(t)|^2\dd t - \sigI(t)\cdot\sigJ(t)\dd t \\[4pt]
  &= \big(r^N_t - r^C_t\big)\dd t + \sigJ(t)\cdot dW^N_t \\
  &\quad\; - \big(r^N_t - r^R_t\big)\dd t - \sigI(t)\cdot dW^N_t \\
  &\quad\; + |\sigI(t)|^2\dd t - \sigI(t)\cdot\sigJ(t)\dd t \\[4pt]
  &= \Big(r^R_t - r^C_t + |\sigI(t)|^2 - \sigI(t)\cdot\sigJ(t)\Big)\dd t + \big(\sigJ(t) - \sigI(t)\big)\cdot dW^N_t.
\end{align*}
The drift terms $|\sigI|^2 - \sigI\cdot\sigJ$ constitute the \emph{It\^o convexity correction} (the Siegel paradox correction) that necessarily arises in the ratio of two lognormal processes.
\end{proof}

\begin{remark}[Consistency check under $\QR$]
Under $\QR$, the Girsanov drift adjustment is $dW^R_t = dW^N_t - \bm{\rho}\,\sigI(t)\dd t$ (\Cref{prop:girsanov}). Substituting $dW^N = dW^R + \bm{\rho}\,\sigI\dd t$ into the diffusion of \eqref{eq:dK}:
\begin{align*}
  \frac{dK_t}{K_t} &= \Big(r^R_t - r^C_t + |\sigI|^2 - \sigI\cdot\sigJ\Big)\dd t + (\sigJ - \sigI)\cdot(dW^R_t + \bm{\rho}\,\sigI\dd t) \\
  &= \Big(r^R_t - r^C_t + |\sigI|^2 - \sigI\cdot\sigJ + (\sigJ - \sigI)\cdot\sigI\Big)\dd t + (\sigJ - \sigI)\cdot dW^R_t,
\end{align*}
where the substitution uses $(\sigJ-\sigI)^\top\bm{\rho}\,\sigI = (\sigJ-\sigI)\cdot\sigI = \sigJ\cdot\sigI - |\sigI|^2$ under the inner-product convention of \Cref{rem:inner}. The drift collapses to $r^R_t - r^C_t$, confirming that $K_t B^C_t / B^R_t$ is a $\QR$-martingale as required by \eqref{eq:mart_KC}.
\end{remark}

\begin{remark}[Simulation practice]
For simulation purposes, one evolves $I_t$ and $J_t$ independently via \eqref{eq:dI} to \eqref{eq:dJ} and computes $K_t = J_t/I_t$ directly. The cross-rate dynamics \eqref{eq:dK} are therefore not needed for implementation but serve as a consistency check and appear in certain pricing formulae involving real-credit cross products.
\end{remark}

\subsection{HJM Dynamics Under the Nominal Measure}
\label{sec:hjm}

\subsubsection{HJM framework within each economy}
\label{subsec:hjm_own}

Under each economy's own risk-neutral measure $\mathbb{Q}^j$, the instantaneous forward rate follows a \citet{heath1992bond} specification:
\begin{equation}\label{eq:hjm_generic}
  df^j(t,T) = \mu^j(t,T)\dd t + \sigvec{j}(t,T)\cdot dW^j_t,
\end{equation}
where $\sigvec{j}(t,T)\in\mathbb{R}^m$ is the forward-rate volatility vector. The HJM no-arbitrage drift restriction under $\mathbb{Q}^j$ is:
\begin{equation}\label{eq:hjm_drift}
  \mu^j(t,T) = \sigvec{j}(t,T)\cdot\int_t^T \sigvec{j}(t,u)\dd u \equiv \sigvec{j}(t,T)\cdot\Avec{j}(t,T),
\end{equation}
where we introduce the shorthand $\Avec{j}(t,T) \equiv \int_t^T \sigvec{j}(t,u)\dd u$ for the cumulative volatility integral.

\subsubsection{Measure changes: the Girsanov connections}
\label{app:girsanov_subsec}

The key to expressing foreign-curve dynamics under $\QN$ is the Girsanov relationship between each foreign Brownian motion and $W^N_t$.

\begin{proposition}[Girsanov connections]
\label{prop:girsanov}
The Brownian motions under the three measures are related by:
\begin{align}
  dW^R_t &= dW^N_t - \bm{\rho}\,\sigI(t)\dd t, \label{eq:girsanov_R} \\[4pt]
  dW^C_t &= dW^N_t - \bm{\rho}\,\sigJ(t)\dd t. \label{eq:girsanov_C}
\end{align}
\end{proposition}

\begin{proof}
The Radon-Nikodym derivative from $\QN$ to $\QR$ is proportional to $I_t B^R_t / B^N_t$. By \eqref{eq:dI}, the stochastic exponential driving this change of measure has volatility $\sigI(t)$. The Girsanov drift correction for $W^N$ under $\QR$ is the cross-variation $d\langle W^N, \sigI\cdot W^N\rangle_t/dt = \bm{\rho}\,\sigI(t)$, giving \eqref{eq:girsanov_R}. The same argument with $(R,I)$ replaced by $(C,J)$ yields \eqref{eq:girsanov_C}. Substituting either relation into an integrand $\sigvec{j}\cdot dW^j$ produces the drift correction $\sigvec{j}\cdot\sigI$ (resp.\ $\sigvec{j}\cdot\sigJ$) under the inner-product convention of \Cref{rem:inner}.
\end{proof}

\subsubsection{Nominal forward curve under \texorpdfstring{$\QN$}{QN}}
\label{subsec:nominal_hjm}

The nominal forward rate is ``domestic'': no measure change is needed. Under $\QN$:
\begin{equation}\label{eq:dfN}
  df^N(t,T) = \sigvec{N}(t,T)\cdot\Avec{N}(t,T)\dd t + \sigvec{N}(t,T)\cdot dW^N_t.
\end{equation}

\subsubsection{Real forward curve under \texorpdfstring{$\QN$}{QN}: drift correction from inflation}
\label{subsec:real_hjm}

\begin{proposition}[Real curve under $\QN$]
\label{prop:real_drift}
The real forward rate under the nominal measure satisfies:
\begin{equation}\label{eq:dfR}
  df^R(t,T) = \Big[\sigvec{R}(t,T)\cdot\Avec{R}(t,T) - \sigvec{R}(t,T)\cdot\sigI(t)\Big]\dd t + \sigvec{R}(t,T)\cdot dW^N_t.
\end{equation}
\end{proposition}

\begin{proof}
Under $\QR$, the real forward rate satisfies HJM with the standard drift \eqref{eq:hjm_drift}:
\[
  df^R(t,T) = \sigvec{R}(t,T)\cdot\Avec{R}(t,T)\dd t + \sigvec{R}(t,T)\cdot dW^R_t.
\]
Substituting the Girsanov connection \eqref{eq:girsanov_R} and using the inner-product convention of \Cref{rem:inner} (so that $\sigvec{R}^\top\bm{\rho}\,\sigI = \sigvec{R}\cdot\sigI$):
\begin{align*}
  df^R(t,T) &= \sigvec{R}(t,T)\cdot\Avec{R}(t,T)\dd t + \sigvec{R}(t,T)\cdot\big(dW^N_t - \bm{\rho}\,\sigI(t)\dd t\big) \\[4pt]
  &= \Big[\sigvec{R}(t,T)\cdot\Avec{R}(t,T) - \sigvec{R}(t,T)\cdot\sigI(t)\Big]\dd t + \sigvec{R}(t,T)\cdot dW^N_t.
\end{align*}
The correction term $-\sigvec{R}(t,T)\cdot\sigI(t)$ is the covariation between the real forward-rate volatility and the inflation exchange-rate volatility. This is the original \citet{jarrow2003hjm} result.
\end{proof}

\subsubsection{Credit forward curve under \texorpdfstring{$\QN$}{QN}: drift correction from credit FX}
\label{subsec:credit_hjm}

\begin{proposition}[Credit curve under $\QN$]
\label{prop:credit_drift}
By exact analogy with \Cref{prop:real_drift}, the credit forward rate under $\QN$ satisfies:
\begin{equation}\label{eq:dfC}
  df^C(t,T) = \Big[\sigvec{C}(t,T)\cdot\Avec{C}(t,T) - \sigvec{C}(t,T)\cdot\sigJ(t)\Big]\dd t + \sigvec{C}(t,T)\cdot dW^N_t.
\end{equation}
\end{proposition}

\begin{proof}
Under $\QC$, the credit forward rate satisfies HJM:
\[
  df^C(t,T) = \sigvec{C}(t,T)\cdot\Avec{C}(t,T)\dd t + \sigvec{C}(t,T)\cdot dW^C_t.
\]
Substituting the Girsanov connection \eqref{eq:girsanov_C} (with $\sigvec{C}^\top\bm{\rho}\,\sigJ = \sigvec{C}\cdot\sigJ$ per \Cref{rem:inner}):
\begin{align*}
  df^C(t,T) &= \sigvec{C}(t,T)\cdot\Avec{C}(t,T)\dd t + \sigvec{C}(t,T)\cdot\big(dW^N_t - \bm{\rho}\,\sigJ(t)\dd t\big) \\[4pt]
  &= \Big[\sigvec{C}(t,T)\cdot\Avec{C}(t,T) - \sigvec{C}(t,T)\cdot\sigJ(t)\Big]\dd t + \sigvec{C}(t,T)\cdot dW^N_t.
\end{align*}
The correction $-\sigvec{C}(t,T)\cdot\sigJ(t)$ is the credit analogue of the inflation correction in \eqref{eq:dfR}: it adjusts the credit-curve drift for the covariation between credit forward-rate innovations and credit-FX innovations.
\end{proof}

\subsubsection{Summary: the three-curve system}
\label{subsec:system}

Collecting \eqref{eq:dfN} to \eqref{eq:dfC}, the full system under $\QN$ is:
\begin{equation}\label{eq:system_app}
  \boxed{
  \begin{aligned}
    df^N(t,T) &= \sigvec{N}(t,T)\cdot\Avec{N}(t,T)\dd t + \sigvec{N}(t,T)\cdot dW^N_t, \\[6pt]
    df^R(t,T) &= \Big[\sigvec{R}(t,T)\cdot\Avec{R}(t,T) - \sigvec{R}(t,T)\cdot\sigI(t)\Big]\dd t + \sigvec{R}(t,T)\cdot dW^N_t, \\[6pt]
    df^C(t,T) &= \Big[\sigvec{C}(t,T)\cdot\Avec{C}(t,T) - \sigvec{C}(t,T)\cdot\sigJ(t)\Big]\dd t + \sigvec{C}(t,T)\cdot dW^N_t.
  \end{aligned}
  }
\end{equation}

All three curves are driven by the \emph{same} $m$-dimensional Brownian motion $W^N_t$. Cross-curve correlations arise through the inner products of the respective volatility vectors $\sigvec{N}$, $\sigvec{R}$, $\sigvec{C}$ within $\mathbb{R}^m$. The structure is symmetric: each foreign curve's drift is its ``own-measure'' HJM drift \emph{minus} a correction equal to the covariance of the forward-rate volatility with the relevant exchange-rate volatility.

\subsection{Bond Prices and the Drift Integrals}
\label{sec:bonds}

For completeness, we record the relationship between bond prices and forward rates under the three-curve system.

\begin{proposition}[Bond price dynamics under $\QN$]
\label{prop:bond}
For each economy $j\in\{N,R,C\}$, the zero-coupon bond price satisfies:
\begin{equation}\label{eq:bond_dyn}
  \frac{dP^j(t,T)}{P^j(t,T)} = r^j_t\dd t - \Avec{j}(t,T)\cdot dW^j_t,
\end{equation}
under $\mathbb{Q}^j$. Under $\QN$, the foreign bond prices (expressed in nominal terms) satisfy:
\begin{align}
  \frac{d\big[I_t P^R(t,T)\big]}{I_t P^R(t,T)} &= r^N_t\dd t + \big(\sigI(t) - \Avec{R}(t,T)\big)\cdot dW^N_t, \label{eq:nom_real_bond} \\[4pt]
  \frac{d\big[J_t P^C(t,T)\big]}{J_t P^C(t,T)} &= r^N_t\dd t + \big(\sigJ(t) - \Avec{C}(t,T)\big)\cdot dW^N_t. \label{eq:nom_credit_bond}
\end{align}
\end{proposition}

\begin{proof}
Equation \eqref{eq:bond_dyn} follows from $P^j(t,T) = \exp\!\big(-\int_t^T f^j(t,u)\dd u\big)$ and the HJM dynamics \eqref{eq:hjm_generic} to \eqref{eq:hjm_drift} by a standard computation \citep[see][Proposition~2]{heath1992bond}.

For \eqref{eq:nom_real_bond}: the nominal value of the real bond is $V_t = I_t P^R(t,T)$. Apply It\^o's product rule using \eqref{eq:dI} and the $\QN$-dynamics of $P^R(t,T)$. The $\QN$-dynamics of $P^R$ are obtained from \eqref{eq:bond_dyn} by the Girsanov substitution $dW^R_t = dW^N_t - \bm{\rho}\,\sigI(t)\dd t$:
\[
  \frac{dP^R(t,T)}{P^R(t,T)} = \Big[r^R_t + \Avec{R}(t,T)\cdot\sigI(t)\Big]\dd t - \Avec{R}(t,T)\cdot dW^N_t.
\]
Then:
\begin{align*}
  \frac{dV_t}{V_t} &= \frac{dI_t}{I_t} + \frac{dP^R}{P^R} + \frac{dI_t}{I_t}\cdot\frac{dP^R}{P^R} \\[4pt]
  &= (r^N_t - r^R_t)\dd t + \sigI\cdot dW^N_t + r^R_t\dd t + \Avec{R}\cdot\sigI\dd t - \Avec{R}\cdot dW^N_t - \sigI\cdot\Avec{R}\dd t \\[4pt]
  &= r^N_t\dd t + (\sigI - \Avec{R})\cdot dW^N_t.
\end{align*}
The covariance terms $+\Avec{R}\cdot\sigI$ and $-\sigI\cdot\Avec{R}$ cancel. Equation \eqref{eq:nom_credit_bond} follows identically.
\end{proof}

\begin{remark}
Equations \eqref{eq:nom_real_bond} to \eqref{eq:nom_credit_bond} confirm that $I_t P^R(t,T)/B^N_t$ and $J_t P^C(t,T)/B^N_t$ are $\QN$-martingales, as required by the general no-arbitrage conditions. These dynamics are needed for pricing nominal instruments with foreign-economy payoffs, e.g., an inflation-protected corporate bond paying in real terms with credit risk.
\end{remark}

\subsection{Spread Dynamics and the Triangle Consistency Identity}
\label{sec:spreads}

\section{Calibration procedure}\label{app:calibration}

The three-curve system of \Cref{sec:framework} is fully specified once a volatility architecture for the forward-rate curves and exchange rates is chosen and a procedure for estimating the parameters from market data is fixed. This appendix develops both. The volatility specification is a parsimonious shared-factor structure that admits closed-form HJM drift integrals; the calibration workflow uses only curve histories and JGP index spread data; the simulation algorithm is a discretised Musiela-grid implementation; and the model outputs are summarised at the end with diagnostic checks.

\subsection{Volatility Architecture}
\label{sec:vol_arch}

\subsubsection{Design requirements}
\label{subsec:vol_requirements}

A practical volatility specification for the three-curve system must satisfy four requirements:
\begin{enumerate}[label=(R\arabic*)]
  \item \label{req:realistic} \textbf{Realistic marginal dynamics.} Each curve (nominal, real, and credit) must exhibit plausible level, slope, and curvature movements.
  \item \label{req:comovement} \textbf{Cross-curve co-movements.} The model must reproduce the empirically observed correlations between nominal rate changes, real rate changes, inflation surprises, and credit spread changes.
  \item \label{req:tractable} \textbf{Analytic tractability.} The HJM drift integrals $\Avec{j}(t,T) = \int_t^T \sigvec{j}(t,u)\dd u$ should be available in closed form to avoid numerical quadrature at each simulation step.
  \item \label{req:calibratable} \textbf{Calibratability from available data.} In the Brazilian market, the typical data are curve histories and index spread histories. No interest-rate option prices (caps, swaptions) are reliably available for the real or credit curves. The volatility architecture must therefore be identifiable from time-series data alone.
\end{enumerate}

\subsubsection{Shared-factor structure}
\label{subsec:shared_factor}

Let $W^N_t \in \mathbb{R}^m$ be the $m$-dimensional Brownian motion under $\QN$ with correlation matrix $\bm{\rho}$ (\Cref{rem:inner}). We decompose each curve's forward-rate volatility using a set of \textbf{shared shape functions} $\{g_q(\tau)\}_{q=1}^{n_g}$, where $\tau = T - t$ is time to maturity, and curve-specific \textbf{amplitude matrices} $A^j(t)\in\mathbb{R}^{m_j\times n_g}$:
\begin{equation}\label{eq:vol_decomp}
  \sigvec{j}(t,T) = \sum_{p=1}^{m_j}\bigg[\sum_{q=1}^{n_g} A^j_{p,q}(t)\,g_q(T-t)\bigg]\mathbf{e}^{(j)}_p, \qquad j \in \{N, R, C\},
\end{equation}
where $\mathbf{e}^{(j)}_p$ denotes the $p$-th canonical basis vector of block $j$ within $\mathbb{R}^m$. Each row $A^j_{p,\cdot}$ specifies how factor $p$ of curve $j$ mixes the basis shapes; the diagonal special case $A^j_{p,q}=a^j_p\,\delta_{p,q}$ (with $n_g=m_j$) recovers a one-shape-per-factor specification. The shape functions are shared across curves, encoding the common maturity structure; the amplitude matrices are curve-specific. We do not impose a sign restriction on $A^j_{p,q}$: the sign reflects the orientation of the corresponding PCA eigenvector and is resolved by the calibration of \Cref{sec:calib_ts}.

\begin{definition}[Exponential volatility family]
\label{def:exp_vol}
We adopt the exponential shape functions:
\begin{equation}\label{eq:shape_fns}
  g_1(\tau) = 1, \qquad g_2(\tau) = e^{-b_2\tau}, \qquad g_3(\tau) = \tau\,e^{-b_3\tau},
\end{equation}
with decay parameters $b_2, b_3 > 0$. These correspond loosely to:
\begin{itemize}[leftmargin=2em]
  \item $g_1$: a \emph{level} factor affecting all maturities equally,
  \item $g_2$: a \emph{slope} factor with exponential decay, concentrated at the short end,
  \item $g_3$: a \emph{curvature} (hump) factor peaking at $\tau^* = 1/b_3$ and decaying thereafter.
\end{itemize}
\end{definition}

The choice of shared decay parameters $b_2, b_3$ across curves is a deliberate parsimony constraint: it forces the ``shape'' of volatility over the maturity axis to be common, while allowing the ``magnitude'' of each factor to differ by curve via the amplitude rows $A^j_{p,\cdot}$. This is empirically well motivated, since PCA analyses of yield-curve changes in different economies and asset classes consistently recover level/slope/curvature factors with similar decay profiles \citep{litterman1991common}.

\begin{proposition}[Closed-form drift integrals]
\label{prop:drift_integrals}
For the shape functions \eqref{eq:shape_fns}, the cumulative volatility integrals $G_q(\tau) = \int_0^\tau g_q(s)\dd s$ are:
\begin{align}
  G_1(\tau) &= \tau, \label{eq:G1_ch4} \\[4pt]
  G_2(\tau) &= \frac{1 - e^{-b_2\tau}}{b_2}, \label{eq:G2_ch4} \\[4pt]
  G_3(\tau) &= \frac{1 - (1 + b_3\tau)\,e^{-b_3\tau}}{b_3^2}. \label{eq:G3_ch4}
\end{align}
Consequently, the cumulative volatility of curve $j$ is $\Avec{j}(t,\tau) = \sum_{p=1}^{m_j}\big[\sum_{r=1}^{n_g} A^j_{p,r}(t)\,G_r(\tau)\big]\mathbf{e}^{(j)}_p$, and the HJM drift under the curve's own measure is the bilinear form
\begin{equation}\label{eq:drift_full_j}
  \mu^j(t,\tau) = \sigvec{j}(t,\tau)\cdot\Avec{j}(t,\tau) = \sum_{q,r=1}^{n_g}\big[A^j(t)^\top A^j(t)\big]_{q,r}\,g_q(\tau)\,G_r(\tau),
\end{equation}
which collapses to $\sum_{p=1}^{m_j} (a^j_p)^2\,g_p(\tau)\,G_p(\tau)$ in the diagonal special case.
\end{proposition}

\begin{proof}
Direct integration using $\int_0^\tau e^{-b s}\dd s = (1-e^{-b\tau})/b$ and $\int_0^\tau s\,e^{-b s}\dd s = [1-(1+b\tau)e^{-b\tau}]/b^2$ yields \eqref{eq:G1_ch4} to \eqref{eq:G3_ch4}. For \eqref{eq:drift_full_j}, the basis vectors of block $j$ are mutually orthogonal under $\bm{\rho}_{jj}=\mathbf{I}_{m_j}$ (\Cref{rem:inner}), so $\sigvec{j}\cdot\Avec{j} = \sum_p \sigvec{j}_p\,A^j_p$ where $\sigvec{j}_p(\tau)=\sum_q A^j_{p,q}\,g_q(\tau)$ and $A^j_p(\tau)=\sum_r A^j_{p,r}\,G_r(\tau)$. Expanding and exchanging the order of summation gives $\sum_{q,r}\big[\sum_p A^j_{p,q} A^j_{p,r}\big] g_q(\tau)\,G_r(\tau) = \sum_{q,r} (A^{j\top} A^j)_{q,r}\, g_q(\tau)\, G_r(\tau)$.
\end{proof}

Requirement \ref{req:tractable} is therefore satisfied: no numerical quadrature is needed within each simulation time step.

\subsubsection{Base-plus-spread decomposition}
\label{subsec:base_spread}

Rather than modeling the credit forward rate $f^C$ directly via \eqref{eq:vol_decomp}, it is empirically more natural, and more stable, to decompose:
\begin{equation}\label{eq:base_spread_decomp}
  f^C(t,T) = f^N(t,T) + s^{CDI}(t,T),
\end{equation}
and to assign separate factor structures to the nominal curve and the CDI+ spread process.

\begin{definition}[Rate and spread factor blocks]
\label{def:factor_blocks}
Partition the $m$-dimensional Brownian motion into three blocks of dimensions $m_N$, $m_R$, and $m_S$ (with $m = m_N + m_R + m_S$). With rate shape functions $\{g_q\}_{q=1}^{n_g}$ and spread shape functions $\{h_q\}_{q=1}^{n_h}$, the block volatility specifications are:
\begin{align}
  \sigvec{N}(t,T) &= \sum_{p=1}^{m_N}\bigg[\sum_{q=1}^{n_g} A^N_{p,q}(t)\,g_q(T-t)\bigg]\mathbf{e}_p, \label{eq:sigN_block} \\[4pt]
  \sigvec{R}(t,T) &= \sum_{p=1}^{m_R}\bigg[\sum_{q=1}^{n_g} A^R_{p,q}(t)\,g_q(T-t)\bigg]\mathbf{e}_{m_N+p}, \label{eq:sigR_block} \\[4pt]
  \sigvec{S}(t,T) &= \sum_{p=1}^{m_S}\bigg[\sum_{q=1}^{n_h} A^S_{p,q}(t)\,h_q(T-t)\bigg]\mathbf{e}_{m_N+m_R+p}, \label{eq:sigS_block}
\end{align}
where $\sigvec{S}$ is the volatility of the CDI+ spread process, the spread shapes $\{h_q\}$ may differ from $\{g_q\}$, and the canonical basis vectors are chosen from non-overlapping blocks. Cross-block correlations enter through the off-diagonal entries of $\bm{\rho}$ (\Cref{rem:inner,sec:calib_corr}).
\end{definition}

The credit-curve volatility is then recovered as:
\begin{equation}\label{eq:sigC_from_decomp}
  \sigvec{C}(t,T) = \sigvec{N}(t,T) + \sigvec{S}(t,T),
\end{equation}
and the IPCA+ spread is implied by the triangle identity (Chapter~3, Proposition~5.2):
\begin{equation}\label{eq:sIPCA_implied}
  s^{IPCA}(t,T) = s^{CDI}(t,T) + f^N(t,T) - f^R(t,T).
\end{equation}

\begin{remark}[Why not model three curves independently?]
\label{rem:why_decomp}
One could in principle assign independent factor structures to $f^N$, $f^R$, and $f^C$ and calibrate all correlations from time-series data. The base-plus-spread decomposition has two advantages. First, it automatically enforces that the credit forward rate exceeds the nominal forward rate when spreads are positive, a constraint that independent modeling could violate under extreme scenarios. Second, the spread process $s^{CDI}$ has substantially lower dimensionality and volatility than $f^C$ itself, leading to more stable PCA decompositions and parameter estimates.
\end{remark}

\begin{remark}[Spread shape functions]
\label{rem:spread_shapes}
For the spread factors $\{h_k\}$, a natural choice mirrors the rate factors:
\begin{equation}\label{eq:spread_shapes}
  h_1(\tau) = 1, \qquad h_2(\tau) = e^{-c_2\tau},
\end{equation}
with $c_2 > 0$. Two spread factors typically suffice: a level factor capturing parallel spread moves and a slope factor capturing steepening/flattening. A curvature factor may be added if the data support it, but in practice credit spread curves are lower-dimensional than rate curves.
\end{remark}

\subsubsection{Exchange-rate volatilities}
\label{subsec:fx_vol_spec}

The inflation and credit exchange rates require volatility specifications compatible with the factor structure:
\begin{equation}\label{eq:sigI_spec}
  \sigI(t) = \sum_{k=1}^m \alpha^I_k(t)\,\mathbf{e}_k, \qquad \sigJ(t) = \sum_{k=1}^m \alpha^J_k(t)\,\mathbf{e}_k,
\end{equation}
where $\alpha^I_k, \alpha^J_k \in\mathbb{R}$ are scalar loadings on each Brownian factor; signs reflect the orientation of the corresponding eigenvector.

We impose the following \textbf{sparsity restrictions} to limit dimensionality while preserving economic interpretability:

\begin{enumerate}[label=(S\arabic*)]
  \item \label{spar:inflation} $\sigI$ loads primarily on the real-curve factor block (indices $m_N+1,\ldots,m_N+m_R$), so that inflation shocks are naturally correlated with the nominal-real differential. A small loading on the nominal level factor may be included.
  \item \label{spar:credit} $\sigJ$ loads primarily on the spread factor block (indices $m_N+m_R+1,\ldots,m$), so that credit-FX shocks are correlated with spread movements. A small loading on the nominal level factor captures the empirical observation that spreads widen during rate rallies.
  \item \label{spar:zero} All other cross-loadings are set to zero as a baseline, yielding a model with $m_R + 1$ free parameters for $\sigI$ and $m_S + 1$ free parameters for $\sigJ$.
\end{enumerate}

\begin{remark}[Equivalent parametrizations]
\label{rem:equiv_param}
An alternative formulation replaces the canonical-basis amplitudes with a non-diagonal correlation matrix $\Sigma$ among the Brownian factors and uses unit loadings. The two representations produce the same joint distribution of curve increments. We adopt the amplitude-based parametrization because it facilitates PCA-based estimation: the amplitude rows $A^j_{p,\cdot}$ are directly interpretable as factor volatilities in units of basis points per square root of time, and the cross-block correlations sit explicitly in $\bm{\rho}$ (\Cref{rem:inner,sec:calib_corr}).
\end{remark}

\subsubsection{Parameter count}
\label{subsec:param_count}

\Cref{tab:params} summarizes the free parameters for a typical specification with $m_N = 3$, $m_R = 2$, $m_S = 2$ (total $m = 7$ Brownian factors), $n_g=3$ rate shape functions, and $n_h=2$ spread shape functions.

\noindent
Of these, the 19 amplitude-matrix entries (9 nominal, 6 real, 4 spread) and 3 decay parameters are identified from PCA followed by NLS in \Cref{sec:calib_ts}, the 6 FX loadings from moment matching in \Cref{sec:calib_fx}, and the correlations from direct estimation in \Cref{sec:calib_corr}. The triangle identity provides out-of-sample validation throughout. In practice, many of the 15 cross-block correlations are negligibly small and may be set to zero, reducing the effective parameter count further; sparsity restrictions \ref{spar:inflation}--\ref{spar:zero} also collapse most of the FX-loading dimensionality (\Cref{subsec:fx_vol_spec}).

\subsection{Calibration Step A: Cross-Sectional Curve Fitting}
\label{sec:calib_cross}

\subsubsection{Inputs}
\label{subsec:inputs}

At each calibration date $t_0$, we require:
\begin{enumerate}[label=(\alph*)]
  \item The \textbf{nominal base curve} $\{P^N(t_0, T_i)\}_{i=1}^{n_N}$, obtained from the DI futures strip via standard bootstrapping.
  \item The \textbf{real base curve} $\{P^R(t_0, T_i)\}_{i=1}^{n_R}$, obtained from NTN-B (IPCA-linked government bond) yields stripped to zero-coupon form using the \citet{svensson1994estimating} parametrization or cubic-spline interpolation.
  \item The \textbf{CDI+ spread term structure} $\{s^{CDI}_{mkt}(t_0, T_i)\}$ from JGP Idex-CDI index data.
  \item The \textbf{IPCA+ spread term structure} $\{s^{IPCA}_{mkt}(t_0, T_i)\}$ from JGP Idex-IPCA index data.
  \item \textbf{Historical time series} of all four objects above, at daily or weekly frequency, over a window of length $L$ (typically $L = 1$ to $3$ years).
\end{enumerate}

Items (a) to (d) are used for cross-sectional fitting at $t_0$; item (e) is used for time-series estimation (Sections~\ref{sec:calib_ts} to \ref{sec:calib_fx}).

\subsubsection{Constructing the credit curve}
\label{subsec:credit_curve}

The credit forward curve is constructed by adding CDI+ spreads to the nominal forwards:
\begin{equation}\label{eq:fC_construct}
  f^C(t_0, T) := f^N(t_0, T) + s^{CDI}_{mkt}(t_0, T).
\end{equation}

This definition is adopted rather than estimated: it is a direct consequence of identification \eqref{eq:base_spread_decomp}.

\subsubsection{Triangle consistency check}
\label{subsec:triangle_check}

Compute the model-implied IPCA+ spread:
\begin{equation}\label{eq:sIPCA_impl}
  s^{IPCA}_{impl}(t_0, T) := f^C(t_0, T) - f^R(t_0, T) = s^{CDI}_{mkt}(t_0, T) + f^N(t_0, T) - f^R(t_0, T),
\end{equation}
and assess the \textbf{triangle consistency gap}:
\begin{equation}\label{eq:gap}
  \Delta(t_0, T) := s^{IPCA}_{impl}(t_0, T) - s^{IPCA}_{mkt}(t_0, T).
\end{equation}

A non-zero $\Delta$ may arise from several sources:
\begin{enumerate}[label=(\roman*)]
  \item \textbf{Base-curve mismatch.} The real curve derived from NTN-B bootstrapping may differ from the IPCA base curve used internally by JGP in constructing its IPCA+ indices. Differences in interpolation method, daycount convention, or the treatment of seasonality in inflation can each contribute several basis points.
  \item \textbf{Index methodology differences.} The CDI+ and IPCA+ JGP indices may use different eligibility screens, weighting schemes, or rebalancing frequencies, so their spread levels need not be perfectly consistent even for the same underlying credit universe.
  \item \textbf{Convexity and liquidity effects.} CDI-linked and IPCA-linked debentures of the same issuer may trade at different effective spreads due to convexity differences (the IPCA linker has a real-rate duration plus an inflation duration) or differential liquidity.
\end{enumerate}

\subsubsection{Deterministic reconciliation adjustment}
\label{subsec:reconciliation}

When $\Delta(t_0, T)$ is structurally non-zero, we introduce a small deterministic adjustment $\chi(T)$:
\begin{equation}\label{eq:adjustment}
  f^C(t_0, T) := f^N(t_0, T) + s^{CDI}_{mkt}(t_0, T) + \chi(T),
\end{equation}
where $\chi$ is calibrated by minimizing a weighted least-squares objective:
\begin{equation}\label{eq:chi_opt}
  \min_{\chi \in \mathcal{C}} \sum_{i=1}^{n} w(T_i)\,\Delta(t_0, T_i)^2,
\end{equation}
subject to $\chi \in \mathcal{C}$, a space of smooth functions (e.g., cubic splines with a small number of knots, or polynomials of degree $\leq 3$).

\begin{remark}[Interpretation and magnitude]
\label{rem:chi_magnitude}
The adjustment $\chi$ absorbs the ``residual'' of the triangle identity in market data. It should be smooth, small in magnitude (typically 2 to 8 basis points), and stable across consecutive calibration dates. A large or volatile $\chi$ signals fundamental data issues, for instance, a mismatched real-curve source, that should be resolved before proceeding to time-series estimation. In production, we recommend tracking $\|\chi\|_\infty$ as a data-quality diagnostic.
\end{remark}

\subsection{Calibration Step B: Time-Series Factor Estimation}
\label{sec:calib_ts}

\subsubsection{Computing historical changes}
\label{subsec:hist_changes}

Over the estimation window $\{t_1, t_2, \ldots, t_L\}$ (daily or weekly dates), compute forward-rate changes at each pillar maturity $T_i$:
\begin{align}
  \Delta f^N(t_\ell, T_i) &= f^N(t_{\ell+1}, T_i) - f^N(t_\ell, T_i), \label{eq:dfN_hist} \\[4pt]
  \Delta f^R(t_\ell, T_i) &= f^R(t_{\ell+1}, T_i) - f^R(t_\ell, T_i), \label{eq:dfR_hist} \\[4pt]
  \Delta s^{CDI}(t_\ell, T_i) &= s^{CDI}(t_{\ell+1}, T_i) - s^{CDI}(t_\ell, T_i). \label{eq:dsS_hist}
\end{align}

Let $\mathbf{X}^j \in \mathbb{R}^{(L-1) \times n_j}$ denote the matrix of changes for block $j \in \{N, R, S\}$, where each row is a date and each column a maturity pillar.

\begin{remark}[Maturity alignment]
\label{rem:maturity_alignment}
A subtlety arises because, on successive dates, the same calendar maturity $T_i$ corresponds to a \emph{decreasing} time-to-maturity $\tau_i = T_i - t_\ell$. In the Musiela parametrization, forward rates are indexed by $\tau$ rather than $T$, and the change $\Delta f^j(t_\ell, \tau_i)$ must be computed at constant $\tau$ (constant time-to-maturity), which requires interpolation. For daily data and maturities beyond one year, the difference is negligible (one business day out of $\geq 252$), but for short maturities it should be handled explicitly.
\end{remark}

\subsubsection{Principal component analysis}
\label{subsec:pca}

For each block $j$, perform the eigendecomposition of the sample covariance matrix of mean-centered weekly changes:
\begin{equation}\label{eq:cov_decomp}
  \hat{\Sigma}^j = \frac{1}{\Delta t\,(L-2)}\,(\tilde{\mathbf{X}}^j)^\top \tilde{\mathbf{X}}^j = \sum_{k=1}^{n_j} \lambda^j_k\,\mathbf{v}^j_k\,(\mathbf{v}^j_k)^\top, \qquad \tilde{\mathbf{X}}^j_{\ell,i} = \mathbf{X}^j_{\ell,i} - \bar{X}^j_i,
\end{equation}
with $\bar X^j_i$ the column mean and $\Delta t = 1/52$~yr the weekly time step that annualizes the variance. The eigenvalues $\lambda^j_1 \geq \lambda^j_2 \geq \cdots \geq 0$ are reported in bp$^2$/yr and the eigenvectors $\mathbf{v}^j_k\in\mathbb{R}^{n_j}$ are unit norm. Empirical weekly increments have mean two orders of magnitude smaller than their standard deviation, so centering is a minor correction in practice; it is retained for conformance with the standard sample-covariance estimator.

We retain factors by parsimony: $m_N=3$, $m_R=2$, $m_S=2$, with the third nominal factor and second spread factor included only if each individually explains $\geq 3\%$ of the block variance. This rule reflects the standard \citet{litterman1991common} level/slope/curvature decomposition for rate blocks and a level-plus-slope decomposition for spread blocks. The cumulative variance shares for the retained set, reported in \Cref{subsec:dual_listed}, range from 82\% (CDI$^+$) to 90.5\% (real). A 95\% retention rule would require an additional fourth nominal factor (4.6\% incremental variance) and third real factor (6.6\%) without obvious term-structure interpretation; reconstruction RMSE varies monotonically with the factor count but none of the perturbations shifts OOS RMSE by more than 2~bp, so the baseline $(3,2,2)$ retention is not material for the headline wedge.

The retained eigenvectors $\mathbf{v}^j_1, \ldots, \mathbf{v}^j_{m_j}$, evaluated at the pillar maturities $\{\tau_i\}$, serve as empirical loading shapes $\sqrt{\lambda^j_k}\,v^j_{k,i}$ that the parametric fit below maps to a curve-specific amplitude matrix $A^j$.

\subsubsection{Parametric fitting of loading shapes}
\label{subsec:loading_fit}

Map the empirical PCA loadings to the parametric family by nonlinear least squares. For the nominal block:
\begin{equation}\label{eq:param_fit}
  \min_{A^N\in\mathbb{R}^{m_N\times n_g},\, b_2,\, b_3} \sum_{p=1}^{m_N} \sum_{i=1}^{n_N} \left[\sqrt{\lambda^N_p}\,v^N_{p,i} - \sum_{q=1}^{n_g} A^N_{p,q}\,g_q(\tau_i)\right]^2,
\end{equation}
where $v^N_{p,i}$ is the $i$-th component (across pillar maturities $\tau_i$) of the $p$-th eigenvector, and $\sqrt{\lambda^N_p}$ converts the unit-norm eigenvector into a volatility-scaled loading. Each row $A^N_{p,\cdot}$ allows factor $p$ to mix the $n_g=3$ basis shapes; the diagonal restriction $A^N_{p,q}=a^N_p\,\delta_{p,q}$ recovers the one-shape-per-factor fit but achieves $R^2$ noticeably below the unrestricted value reported in \Cref{tab:params_nominal}, so the unrestricted form is adopted in the calibration.

\begin{proposition}[Identification of decay parameters]
\label{prop:identification}
The decay parameters $b_2$ and $b_3$ are identified from the maturity profile of the second and third PCA loadings, respectively:
\begin{enumerate}[label=(\roman*)]
  \item $b_2$ controls the rate at which the slope factor decays with maturity. If the empirical second loading has half-life $\tau_{1/2}$ (i.e., $v^N_{2,i}$ reaches half its short-end value at $\tau_{1/2}$), then $b_2 \approx \ln 2 / \tau_{1/2}$.
  \item $b_3$ controls the peak location of the curvature factor: $g_3(\tau)$ peaks at $\tau^* = 1/b_3$. Matching $\tau^*$ to the empirical peak of the third PCA loading identifies $b_3$.
\end{enumerate}
These initial estimates serve as starting values for the joint optimization \eqref{eq:param_fit}.
\end{proposition}

The same procedure is applied to the real curve (with possibly shared $b_2, b_3$ if fitting quality permits) and to the spread curve (with decay parameter $c_2$).

\begin{remark}[Non-parametric alternative]
\label{rem:nonparametric}
If the parametric fit is unsatisfactory, for instance, if the empirical loading shapes deviate substantially from the exponential family, the PCA loadings may be retained directly as spline-interpolated functions of $\tau$. This forfeits the closed-form drift integrals of \Cref{prop:drift_integrals}, requiring numerical integration of $\Avec{j}(t,\tau) = \int_0^\tau \sigvec{j}(t,u)\dd u$ at each step. In our experience with Brazilian DI and real curves, the exponential family provides an adequate fit with $m_N = 3$ and $m_R = 2$ factors.
\end{remark}

\begin{figure}[ht]
\centering
\includegraphics[width=0.48\linewidth]{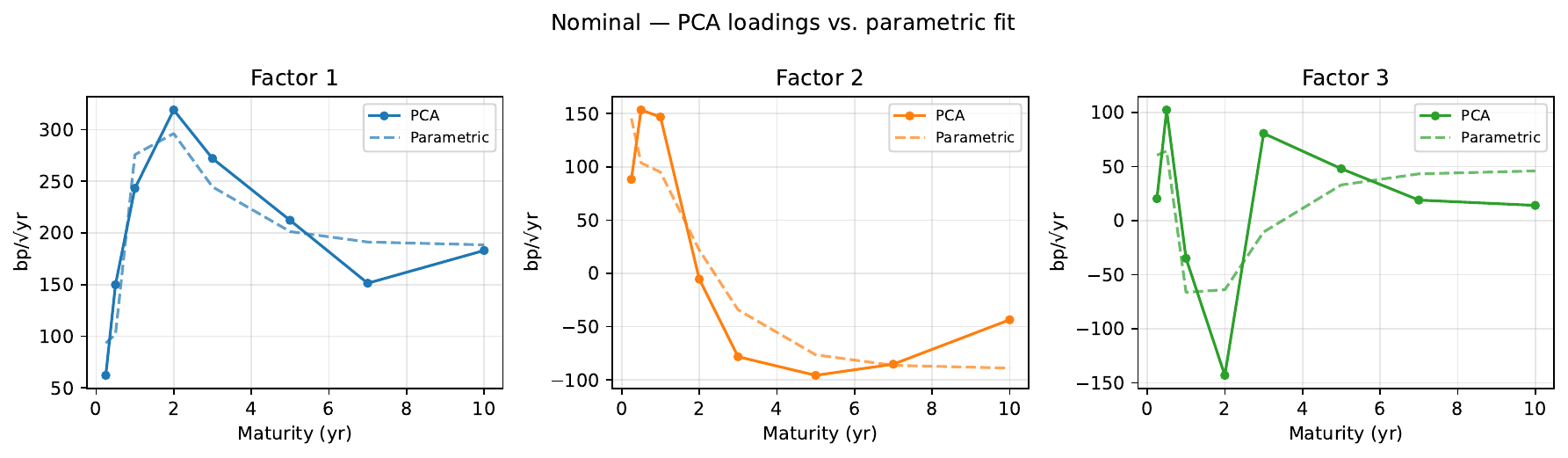}\hfill
\includegraphics[width=0.48\linewidth]{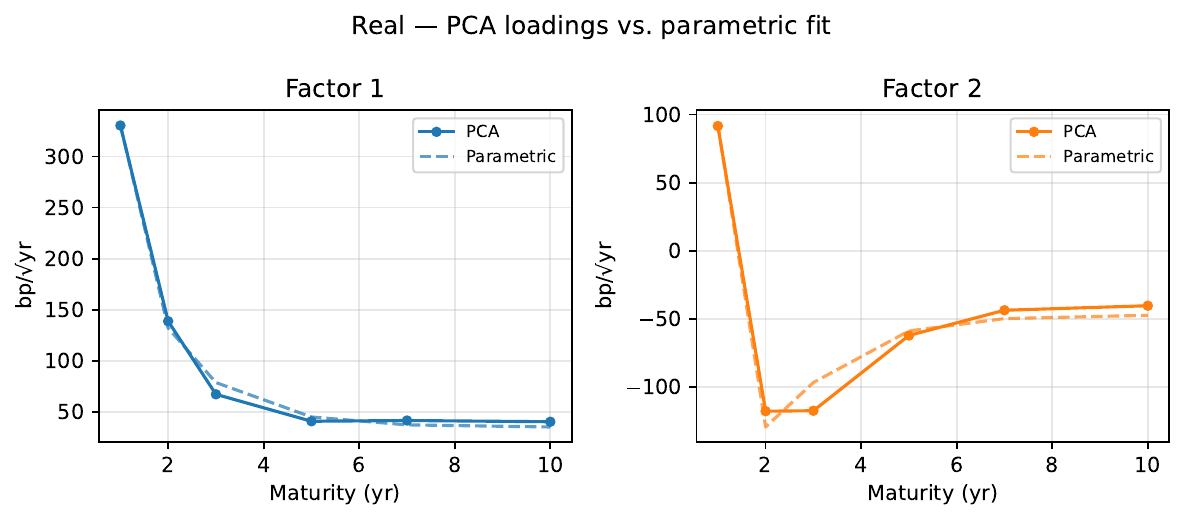}

\medskip

\includegraphics[width=0.48\linewidth]{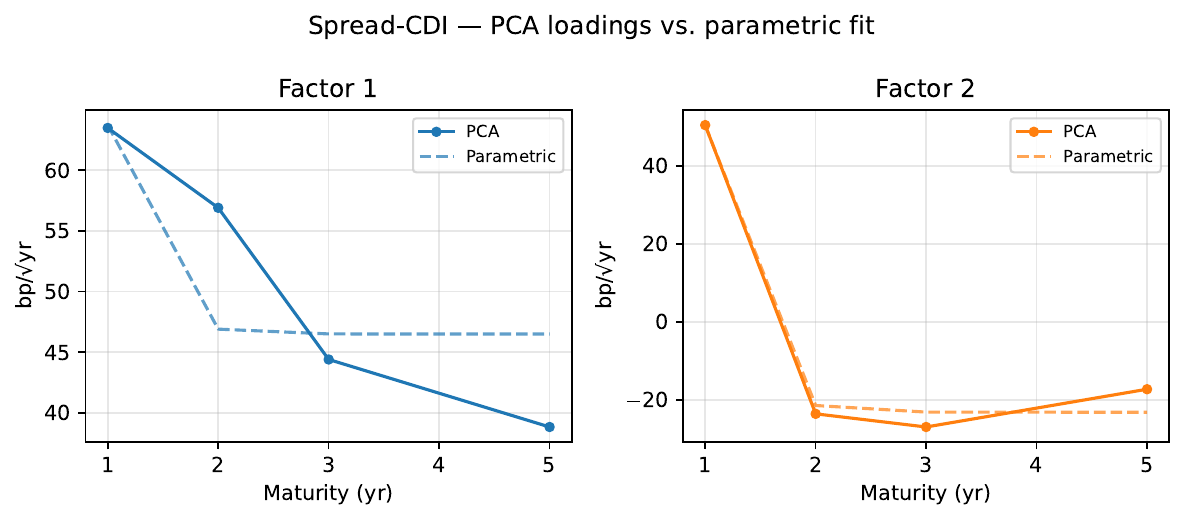}\hfill
\includegraphics[width=0.48\linewidth]{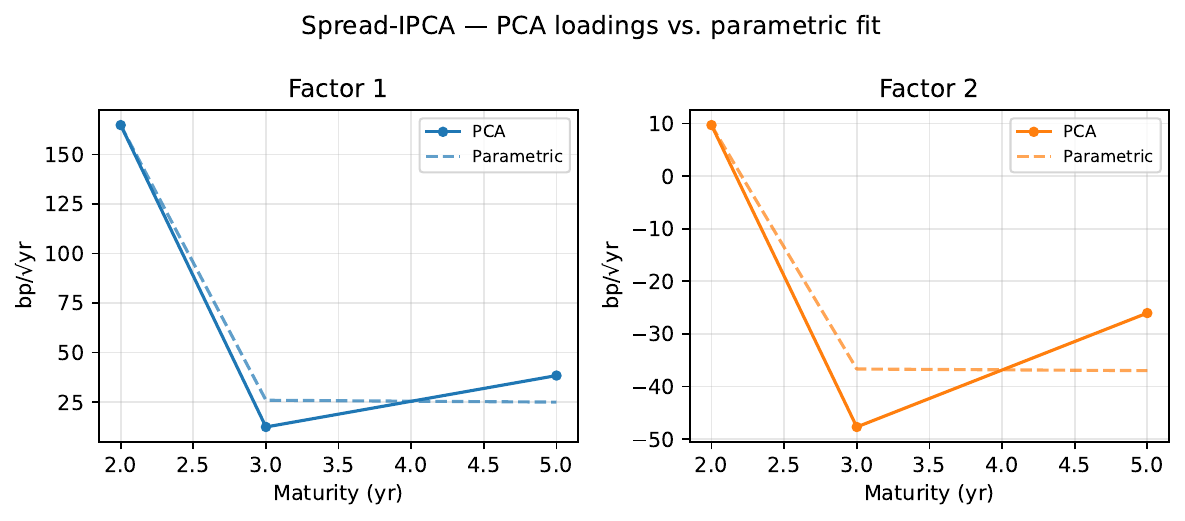}
\caption{Empirical PCA loadings $\sqrt{\lambda^j_k}\,v^j_k$ for the four factor blocks: nominal forwards (top-left), real forwards (top-right), CDI-linked credit spread (bottom-left), and IPCA-linked credit spread (bottom-right). Each panel shows the retained eigenvectors as functions of pillar maturity, after eigendecomposition of the covariance of weekly forward-rate changes over the 2021-01-04 to 2023-12-29 estimation window. Loadings are reported in bp/$\sqrt{\text{yr}}$. The nominal and real blocks exhibit the standard level/slope/curvature shape pattern; the spread blocks decompose into a level factor and a residual short-end factor.}
\label{fig:pca_loadings}
\end{figure}

\subsection{Calibration Step C: Cross-Block Correlation Estimation}
\label{sec:calib_corr}

\subsubsection{Factor innovations}
\label{subsec:factor_innovations}

Project the historical changes onto the retained principal components to obtain factor innovation series. For each block $j$ and date $t_\ell$:
\begin{equation}\label{eq:factor_innovations}
  \Delta Z^j_k(t_\ell) = \sum_{i=1}^{n_j} v^j_{k,i}\,\Delta X^j(t_\ell, \tau_i), \qquad k = 1, \ldots, m_j,
\end{equation}
where $\Delta X^j(t_\ell, \tau_i)$ denotes the forward-rate (or spread) change for block $j$ at maturity $\tau_i$ on date $t_\ell$.

\subsubsection{Joint covariance and correlation}
\label{subsec:joint_cov}

Stack the factor innovations into a combined vector:
\begin{equation}\label{eq:stacked}
  \mathbf{Z}(t_\ell) = \Big(\Delta Z^N_1, \ldots, \Delta Z^N_{m_N},\;\; \Delta Z^R_1, \ldots, \Delta Z^R_{m_R},\;\; \Delta Z^S_1, \ldots, \Delta Z^S_{m_S}\Big)^\top \in \mathbb{R}^m.
\end{equation}

Compute the sample covariance matrix:
\begin{equation}\label{eq:sample_cov}
  \hat{\Sigma} = \frac{1}{L-2}\sum_{\ell=1}^{L-1} \mathbf{Z}(t_\ell)\,\mathbf{Z}(t_\ell)^\top,
\end{equation}
and extract the correlation matrix $\hat{\rho}$ via standardization:
\begin{equation}\label{eq:corr_matrix}
  \hat{\rho}_{ij} = \frac{\hat{\Sigma}_{ij}}{\sqrt{\hat{\Sigma}_{ii}\,\hat{\Sigma}_{jj}}}.
\end{equation}

\subsubsection{Positive-definiteness correction}
\label{subsec:pd_correction}

If $\hat{\rho}$ is not positive definite, which can occur with short estimation windows or near-collinear factors, we apply the \citet{higham2002computing} alternating-projections algorithm to compute the nearest positive-definite correlation matrix:
\begin{equation}\label{eq:higham}
  \rho^* = \arg\min_{\rho \succ 0,\;\text{diag}(\rho)=\mathbf{1}} \|\rho - \hat{\rho}\|_F,
\end{equation}
where $\|\cdot\|_F$ denotes the Frobenius norm. The corrected matrix $\rho^*$ is used for simulation.

\subsubsection{Economically important correlations}
\label{subsec:econ_corr}

Several cross-block correlations have direct economic content and serve as sanity checks:
\begin{enumerate}[label=(\roman*)]
  \item \textbf{Rate-spread correlation.} The correlation between the nominal level factor $\Delta Z^N_1$ and the spread level factor $\Delta Z^S_1$ captures the ``flight to quality'' effect: negative correlation indicates that spreads widen when rates fall. In the Brazilian market, this correlation is typically between $-0.2$ and $-0.5$ during stress episodes and near zero during calm periods.
  \item \textbf{Nominal-real correlation.} The correlation between nominal and real level factors reflects the decomposition of rate moves into real-rate and inflation components. In Brazil, this is typically high ($> 0.7$) at the short end (where both curves are anchored by the Selic rate) and lower at longer maturities.
  \item \textbf{Real-spread correlation.} A negative correlation here implies that IPCA+ debenture spreads tighten when real rates rise, consistent with a ``reaching for yield'' dynamic in which investors accept lower credit compensation when real rates are attractive.
\end{enumerate}

\subsection{Calibration Step D: Exchange-Rate Volatility Estimation}
\label{sec:calib_fx}

Under the two-model framing (Chapter~3, Remark~3.5), the inflation FX $\sigma_I$ is a shared parameter, calibrated once from the realized dynamics of the common inflation index $I_t$. The credit FX is model-specific: Model A has $\sigma_J^A$ (associated with $J^A$, the CDI$^+$ credit-economy exchange rate), and Model B has $\sigma_J^B$ (associated with $J^B$, the IPCA$^+$ credit-economy exchange rate). Neither $\sigma_I$ nor $\sigma_J$ is observable directly, but each can be identified from realized volatility of an appropriate primary object.

\subsubsection{Inflation FX volatility $\sigma_I$ (shared)}
\label{subsec:sigI_calib}

Under $\mathbb{Q}^N$, $dI_t/I_t = (r^N_t - r^R_t)\,dt + \sigma_I \cdot dW^N_t$. Because Girsanov preserves diffusion coefficients, $|\sigma_I|^2$ equals the instantaneous variance of $d\ln I_t$ under the physical measure as well. We therefore identify $|\sigma_I|$ from realized monthly IPCA log-returns:
\begin{enumerate}[label=\textbf{D1.\arabic*}, leftmargin=2.5em]
  \item Collect monthly IPCA values $I_\ell$ from Banco Central do Brasil SGS series 433, align each observation to its reference month (correcting for the $\sim$15-day publication lag), and compute log-returns $r_\ell = \ln(1 + \pi_\ell)$.
  \item De-seasonalize $r_\ell$ by subtracting calendar-month means.
  \item Annualize the variance:
  \[
    |\widehat{\sigma_I}|^2 \;=\; 12 \cdot \widehat{\mathrm{Var}}(r_{\ell,\text{deseas}}).
  \]
  \item Identify the direction $\alpha^I \in \mathbb{R}^m$ by OLS regression of the deseasonalized monthly residuals on monthly-aggregated unit-variance factor innovations (summed weekly PC scores within each calendar month). The regression R$^2$ reports the fraction of realized inflation variance spanned by rate factors; the residual vol is the genuinely idiosyncratic inflation-specific component.
\end{enumerate}

This replaces the ``breakeven residual'' procedure of earlier drafts, under which $\sigma_I$ does not appear in the diffusion of $\beta = f^N - f^R$: that residual does not identify $|\sigma_I|^2$ in the model.

\subsubsection{Credit FX volatility (model-specific)}
\label{subsec:sigJ_calib}

Unlike $\sigma_I$, neither $J^A_t$ nor $J^B_t$ has a directly observable spot proxy, and $\sigma_J$ enters model observables only through the drift correction $-\sigvec{C}\cdot\sigJ$, which is $O(\Delta t)$ in the mean and $O(\Delta t^2)$ in the covariance. $\sigma_J$ is therefore structurally unidentified from curve-plus-spread covariance: the diffusion of $\Delta s^{CDI}$ and $\Delta s^{IPCA}$ is a linear combination of $\sigvec{N}, \sigvec{R}, \sigvec{S}$, none of which contains $\sigJ$. We adopt a \emph{structural identification convention} analogous to the Jarrow-Yildirim treatment of inflation: equate $|\sigJ|$ to the annualized realized vol of the corresponding short-pillar spread series,
\[
  |\widehat{\sigma_J^{A}}| \;:=\; \sqrt{52}\,\cdot\,\widehat{\mathrm{std}}\!\big(\Delta s^{CDI}(t, \tau_{\min})\big),
  \qquad
  |\widehat{\sigma_J^{B}}| \;:=\; \sqrt{52}\,\cdot\,\widehat{\mathrm{std}}\!\big(\Delta s^{IPCA}(t, \tau_{\min})\big),
\]
with the direction $\alpha^J$ obtained by OLS regression of the short-pillar spread change on unit-variance factor innovations. Under \ref{spar:credit}, most of the realized variance should load on the $S$-block factors, with a small residual N-factor loading capturing flight-to-quality co-movement.

\begin{remark}[$\sigma_J$ is set by convention, not estimated]
\label{rem:sigJ_convention}
\eqref{eq:sigI_spec} and the convention above pin $|\sigJ|$ as a model parameter, but the procedure is identification by convention, not estimation from data the model can discriminate. The numerical value affects pricing of credit-FX-sensitive contingent claims (e.g., quanto credit derivatives) but does not change the simulated $f^N$, $f^R$, or $s^{CDI}$ paths, nor the within-issuer wedge of \Cref{subsec:dual_listed}: those depend only on $\sigvec{N}, \sigvec{R}, \sigvec{S}, \sigI$. Should listed credit-derivative prices on JGP-universe names become available, $\sigJ$ could be estimated rather than imposed; the convention places it within the bounds suggested by short-pillar realized volatility.
\end{remark}

\begin{remark}[Why not covariance-matching]
\label{rem:why_not_cov_match}
An earlier draft proposed matching the model-implied $(\Delta s^{CDI}, \Delta s^{IPCA})$ covariance to its sample analog. That objective is degenerate in $\sigJ$ for the reason above: the spread covariance does not contain $\sigJ$. The realized-vol convention is the only route that pins $|\sigJ|$ to a finite value without invoking additional data (e.g., a CDS on a JGP-universe name).
\end{remark}

\subsection{Calibration Step E: Regularization and Stability}
\label{sec:regularization}

To prevent overfitting and ensure robust out-of-sample behavior, we impose the following regularization measures.

\subsubsection{Dimensionality control}
\label{subsec:dim_control}

The total number of factors $m = m_N + m_R + m_S$ should be kept small, typically $5 \leq m \leq 8$. A practical guideline is:
\begin{equation}\label{eq:dim_guide}
  m_N = 3, \qquad m_R = 2, \qquad m_S = 2,
\end{equation}
yielding $m = 7$. The third nominal factor (curvature) and the second spread factor (spread slope) are included only if they individually explain $\geq 3$\% of their block's variance; otherwise they are dropped.

\subsubsection{Shared decay parameters}
\label{subsec:shared_decay}

The decay parameters $b_2, b_3$ are shared across the nominal and real curve blocks. This constraint is motivated by the empirical observation that level/slope/curvature factor shapes are similar across nominal and real yield curves. It reduces the parameter count by 2 and improves estimation stability, particularly for the real curve where the pillar grid $\{T_i\}$ from NTN-B maturities is coarser than the DI strip.

\subsubsection{Loading smoothness penalty}
\label{subsec:smoothness}

When the parametric fit \eqref{eq:param_fit} is replaced by a non-parametric (spline-interpolated) loading function, we add a roughness penalty:
\begin{equation}\label{eq:smoothness_penalty}
  \min \sum_{k,i} \left[\sqrt{\lambda^j_k}\,v^j_{k,i} - \hat{\sigma}^j_k(\tau_i)\right]^2 + \eta \sum_k \int_0^{\tau_{max}} \left[\frac{d^2\hat{\sigma}^j_k(\tau)}{d\tau^2}\right]^2 d\tau,
\end{equation}
where $\eta > 0$ is a smoothness parameter chosen by leave-one-out cross-validation or the generalized cross-validation criterion \citep{wahba1990spline}.

\subsubsection{Out-of-sample validation via the triangle}
\label{subsec:oos_validation}

The triangle identity provides a natural out-of-sample check. At each calibration date in a hold-out period:
\begin{enumerate}[label=(\roman*)]
  \item Compute model-implied $s^{IPCA}_{model}(t, T) = s^{CDI}_{model}(t, T) + f^N_{model}(t, T) - f^R_{model}(t, T)$.
  \item Compare against observed $s^{IPCA}_{mkt}(t, T)$.
  \item Track the root-mean-square error:
  \begin{equation}\label{eq:rmse_triangle}
    \text{RMSE}_{triangle}(t) = \sqrt{\frac{1}{n}\sum_{i=1}^{n}\left[s^{IPCA}_{model}(t, T_i) - s^{IPCA}_{mkt}(t, T_i)\right]^2}.
  \end{equation}
\end{enumerate}

A well-calibrated model should produce $\text{RMSE}_{triangle}$ of the same order as the deterministic adjustment $\chi$ from Step~A (typically 2 to 8 basis points). Systematic deterioration signals model misspecification or parameter instability and should trigger recalibration.

\subsection{Simulation Algorithm}
\label{sec:simulation}

We discretize the three-curve system on a Musiela time-to-maturity grid. Let $\{\tau_0, \tau_1, \ldots, \tau_{N_\tau}\}$ be the maturity grid (with $\tau_0 = 0$ corresponding to the short rate) and $\Delta t$ the simulation time step.

\subsubsection{Musiela transport}
\label{subsec:musiela}

In the Musiela parametrization, forward rates are functions of $(t, \tau)$ where $\tau = T - t$. As time advances by $\Delta t$, each maturity point ``ages'': the maturity that was at $\tau$ is now at $\tau - \Delta t$. This transport effect must be accounted for at each step by interpolating the previous curve onto the shifted grid:
\begin{equation}\label{eq:transport}
  f^j(t, \tau - \Delta t) \approx f^j(t, \tau) - \frac{\partial f^j}{\partial \tau}(t, \tau)\,\Delta t,
\end{equation}
or, more robustly, by linear interpolation between adjacent grid points. The transport term is the continuous-time analogue of the roll-down effect familiar to fixed-income practitioners.

\subsubsection{Discretized dynamics}
\label{subsec:discrete_dynamics}

At each step $t \to t + \Delta t$:

\medskip
\noindent\textbf{Step 1. Generate correlated Brownian increments.}
\begin{equation}\label{eq:bm_gen}
  \Delta W^N = L\,\xi\,\sqrt{\Delta t}, \qquad \xi \sim \mathcal{N}(\mathbf{0}, \mathbf{I}_m),
\end{equation}
where $L$ is the Cholesky factor of the correlation matrix: $\rho^* = L L^\top$.

\medskip
\noindent\textbf{Step 2. Update nominal forward rates.}
For each maturity grid point $\tau_i$:
\begin{equation}\label{eq:update_N}
  f^N(t+\Delta t, \tau_i) = f^N\!(t, \tau_i + \Delta t) + \mu^N(t,\tau_i)\,\Delta t + \sigvec{N}(t,\tau_i)\cdot\Delta W^N,
\end{equation}
where $\mu^N(t,\tau_i) = \sum_{q,r=1}^{n_g} (A^{N\top} A^N)_{q,r}\,g_q(\tau_i)\,G_r(\tau_i)$ from \eqref{eq:drift_full_j}, and $f^N\!(t, \tau_i + \Delta t)$ is obtained by interpolation (the Musiela transport term).

\medskip
\noindent\textbf{Step 3. Update real forward rates} (with inflation drift correction).
\begin{equation}\label{eq:update_R}
  f^R(t+\Delta t, \tau_i) = f^R\!(t, \tau_i + \Delta t) + \Big[\mu^R(t,\tau_i) - \sigvec{R}(t,\tau_i)\cdot\sigI(t)\Big]\,\Delta t + \sigvec{R}(t,\tau_i)\cdot\Delta W^N,
\end{equation}
where $\mu^R(t,\tau_i) = \sum_{q,r=1}^{n_g} (A^{R\top} A^R)_{q,r}\,g_q(\tau_i)\,G_r(\tau_i)$.

\medskip
\noindent\textbf{Step 4. Update credit spread.}

If using the base-plus-spread decomposition:
\begin{equation}\label{eq:update_S}
  s^{CDI}(t+\Delta t, \tau_i) = s^{CDI}\!(t, \tau_i + \Delta t) + \mu^S(t,\tau_i)\,\Delta t + \sigvec{S}(t,\tau_i)\cdot\Delta W^N,
\end{equation}
where $\mu^S$ is the spread drift derived in \Cref{app:spread_drift}, eq.~\eqref{eq:spread_drift} (it includes the within-block term $\sigvec{S}\cdot\Avec{S}$, the cross-block terms $\sigvec{N}\cdot\Avec{S} + \sigvec{S}\cdot\Avec{N}$ proportional to $\bm{\rho}_{NS}$, and the credit-FX correction $-\sigvec{C}\cdot\sigJ$). The credit forward rate is then:
\begin{equation}\label{eq:fC_recover}
  f^C(t+\Delta t, \tau_i) = f^N(t+\Delta t, \tau_i) + s^{CDI}(t+\Delta t, \tau_i).
\end{equation}

\medskip
\noindent\textbf{Step 5. Update exchange rates} (frozen-rate exponential step).
\begin{align}
  I_{t+\Delta t} &= I_t\exp\!\left[\Big(r^N_t - r^R_t - \tfrac{1}{2}|\sigI(t)|^2\Big)\Delta t + \sigI(t)\cdot\Delta W^N\right], \label{eq:update_I} \\[4pt]
  J_{t+\Delta t} &= J_t\exp\!\left[\Big(r^N_t - r^C_t - \tfrac{1}{2}|\sigJ(t)|^2\Big)\Delta t + \sigJ(t)\cdot\Delta W^N\right], \label{eq:update_J} \\[4pt]
  K_{t+\Delta t} &= \frac{J_{t+\Delta t}}{I_{t+\Delta t}}. \label{eq:update_K}
\end{align}

Equations \eqref{eq:update_I}--\eqref{eq:update_J} replace the integrated short-rate differential $\int_t^{t+\Delta t}(r^N_s - r^R_s)\,ds$ with its frozen-coefficient evaluation $(r^N_t - r^R_t)\Delta t$; the step is exact in the lognormal sense conditional on this freezing, and preserves positivity of the exchange rates without the discretization bias of a direct Euler--Maruyama scheme on $I_t$ and $J_t$.

\medskip
\noindent\textbf{Step 6. Extract short rates and implied spreads.}
\begin{align}
  r^j_t &= f^j(t, 0), \qquad j \in \{N, R, C\}, \label{eq:short_rate} \\[4pt]
  s^{IPCA}(t, \tau_i) &= f^C(t, \tau_i) - f^R(t, \tau_i) = s^{CDI}(t,\tau_i) + f^N(t,\tau_i) - f^R(t,\tau_i). \label{eq:sIPCA_sim}
\end{align}

\subsubsection{Pseudocode}
\label{subsec:pseudocode}

\Cref{alg:simulation} presents the complete simulation loop in pseudocode.

\begin{algorithm}[ht]
\caption{Three-Currency HJM Simulation}
\label{alg:simulation}
\begin{algorithmic}[1]
\Require Initial curves $f^N_0(\tau),\,f^R_0(\tau),\,s^{CDI}_0(\tau)$; initial FX $I_0,\,J_0$
\Require Amplitude matrices $A^N, A^R, A^S$; decay parameters $b_2,b_3,c_2$; FX loadings $\sigI,\,\sigJ$
\Require Correlation Cholesky factor $L$ with $LL^\top = \bm{\rho}$; time step $\Delta t$; horizon $T_{sim}$
\Ensure Paths of $\{f^N, f^R, f^C, s^{CDI}, s^{IPCA}, I, J, K\}$
\State $N_{steps} \gets T_{sim}/\Delta t$
\For{$n = 0$ to $N_{steps}-1$}
  \State $t \gets n\,\Delta t$
  \State $r^N_n \gets f^N_n(0),\;\; r^R_n \gets f^R_n(0),\;\; r^C_n \gets f^N_n(0) + s^{CDI}_n(0)$ \Comment{Pre-step short rates}
  \State Draw $\xi \sim \mathcal{N}(\mathbf{0}, \mathbf{I}_m)$
  \State $\Delta W \gets L\,\xi\,\sqrt{\Delta t}$ \Comment{Correlated increments, $\mathrm{Cov}(\Delta W)=\bm{\rho}\,\Delta t$}
  \For{each maturity $\tau_i$ in grid}
    \State $f^N_{n+1}(\tau_i) \gets \textsc{Transport}(f^N_n, \tau_i, \Delta t) + \mu^N(\tau_i)\,\Delta t + \sigvec{N}(\tau_i)\cdot\Delta W$
    \State $f^R_{n+1}(\tau_i) \gets \textsc{Transport}(f^R_n, \tau_i, \Delta t) + [\mu^R(\tau_i) - \sigvec{R}(\tau_i)\cdot\sigI]\,\Delta t + \sigvec{R}(\tau_i)\cdot\Delta W$
    \State $s^{CDI}_{n+1}(\tau_i) \gets \textsc{Transport}(s^{CDI}_n, \tau_i, \Delta t) + \mu^S(\tau_i)\,\Delta t + \sigvec{S}(\tau_i)\cdot\Delta W$
  \EndFor
  \State $f^C_{n+1} \gets f^N_{n+1} + s^{CDI}_{n+1}$ \Comment{Recover credit curve}
  \State $I_{n+1} \gets I_n\exp\!\big[(r^N_n - r^R_n - \tfrac{1}{2}|\sigI|^2)\Delta t + \sigI\cdot\Delta W\big]$
  \State $J_{n+1} \gets J_n\exp\!\big[(r^N_n - r^C_n - \tfrac{1}{2}|\sigJ|^2)\Delta t + \sigJ\cdot\Delta W\big]$
  \State $K_{n+1} \gets J_{n+1}/I_{n+1}$
  \State $s^{IPCA}_{n+1} \gets s^{CDI}_{n+1} + f^N_{n+1} - f^R_{n+1}$ \Comment{Triangle identity}
\EndFor
\end{algorithmic}
\end{algorithm}

\subsection{Model Outputs and Diagnostics}
\label{sec:outputs}

A single run of \Cref{alg:simulation} produces, at each simulated date, the following objects:

\begin{enumerate}[label=(\roman*)]
  \item \textbf{Nominal forward-rate curve} $f^N(t, \tau)$ and the corresponding zero-coupon curve $P^N(t, T)$.
  \item \textbf{Real forward-rate curve} $f^R(t, \tau)$ and the corresponding real zero-coupon curve $P^R(t, T)$.
  \item \textbf{Inflation index} $I_t$, from which period-over-period inflation rates can be computed.
  \item \textbf{Corporate credit forward-rate curve} $f^C(t, \tau)$ and the credit zero-coupon curve $P^C(t, T)$.
  \item \textbf{CDI+ credit spread curve} $s^{CDI}(t, \tau) = f^C(t,\tau) - f^N(t,\tau)$.
  \item \textbf{IPCA+ credit spread curve} $s^{IPCA}(t, \tau) = f^C(t,\tau) - f^R(t,\tau)$.
  \item \textbf{Credit FX} $J_t$ and the \textbf{cross rate} $K_t = J_t / I_t$.
\end{enumerate}

\subsubsection{Diagnostic checks}
\label{subsec:diagnostics}

We recommend the following diagnostics to validate simulation quality:

\medskip
\noindent\textbf{D1. Initial curve reproduction.} At $t = 0$, the simulated curves must exactly equal the calibrated initial curves. This is guaranteed by construction of the HJM framework but should be verified numerically to catch implementation errors.

\medskip
\noindent\textbf{D2. Triangle consistency.} At every simulated step and maturity, verify:
\begin{equation}\label{eq:diag_triangle}
  \left|s^{IPCA}(t, \tau) - s^{CDI}(t, \tau) - \big[f^N(t,\tau) - f^R(t,\tau)\big]\right| < \epsilon,
\end{equation}
for a tight numerical tolerance $\epsilon$ (e.g., $10^{-10}$ in continuous-compounding terms). Any violation indicates a coding error, since the triangle is enforced algebraically by \eqref{eq:sIPCA_sim}.

\medskip
\noindent\textbf{D3. Martingale test.} Verify that the deflated foreign numeraires are martingales. For a set of $N_{sim}$ simulation paths, the sample mean of $I_T B^R_T / B^N_T$ at each horizon $T$ should be statistically indistinguishable from $I_0$:
\begin{equation}\label{eq:mart_test}
  \left|\frac{1}{N_{sim}}\sum_{n=1}^{N_{sim}}\frac{I^{(n)}_T B^{R,(n)}_T}{B^{N,(n)}_T} - I_0\right| \leq z_{\alpha/2}\,\frac{\hat{\sigma}}{\sqrt{N_{sim}}},
\end{equation}
where $z_{\alpha/2}$ is the standard normal critical value and $\hat{\sigma}$ is the sample standard deviation across paths. The analogous test applies to $J_T B^C_T / B^N_T$.

\medskip
\noindent\textbf{D4. Marginal statistics.} Compare simulated marginal distributions against historical data:
\begin{enumerate}[label=(\alph*)]
  \item Annualized volatility of forward-rate changes at key maturities (1y, 3y, 5y, 10y) for each curve.
  \item Annualized volatility of CDI+ and IPCA+ spread changes.
  \item Cross-correlations between rate and spread factor innovations.
  \item Distribution of simulated inflation rates versus historical CPI data.
\end{enumerate}

\medskip
\noindent\textbf{D5. Term structure smoothness.} At each simulated date, the forward-rate curves and spread curves should be smooth functions of maturity. Excessive oscillation, particularly at the long end of the grid, signals numerical instability in the Musiela transport step and may require grid refinement or smoothing.

\subsubsection{Computational considerations}
\label{subsec:computational}

The per-step computational cost is $\mathcal{O}(m\,N_\tau)$: for each of the $N_\tau$ maturity grid points, we compute $m$ factor contributions (each involving a closed-form $g_k(\tau)\,G_k(\tau)$ evaluation) and $m$ Brownian increments. The Cholesky factorization of $\rho^*$ is performed once and costs $\mathcal{O}(m^3)$. For $m = 7$ and $N_\tau = 50$, the per-step cost is negligible; a 10-year daily simulation ($\approx 2{,}520$ steps) with $10{,}000$ paths completes in seconds on commodity hardware.

Memory requirements are modest: each path stores $3 \times N_\tau$ forward-rate values plus $3$ exchange-rate values at each step, totaling $\mathcal{O}(N_{steps} \times N_\tau)$ per path. For the baseline parameters, a single path occupies approximately 4~MB of double-precision storage over a 10-year horizon.

\newpage

\section{Out-of-sample diagnostics}\label{app:oos}

All diagnostics use the out-of-sample window (2024-01-02 to 2026-02-05, 110 weekly observations), which was excluded from parameter estimation. Under the two-model framing the within-model triangle identity holds by construction at every simulation step, so there is no ``triangle RMSE'' to track. The diagnostics instead evaluate whether calibrated volatilities match OOS realized changes, whether in-sample correlations persist, and whether the 90\% Gaussian prediction intervals achieve their nominal coverage.

\subsection{Volatility reproduction}\label{subsec:oos_vol}

The rate blocks reproduce OOS realized volatility reasonably well: nominal is 80 to 90\% of model magnitude at short/medium maturities and real is within 10\% at all tenors, reflecting that the OOS window (2024-01 onward) was indeed calmer than the estimation window (which straddled the Selic peak). CDI$^+$ spread vol at the 1-year pillar \emph{underpredicts} the OOS realization by 30\%, traceable to the 2024 credit-fund redemption episode discussed below, while the 3- and 5-year pillars are in line with model.

\subsection{Correlation reproduction}\label{subsec:oos_corr}

The rate-block correlations are stable out-of-sample (difference $\leq 0.02$). The rate-spread correlation flips sign (from weakly negative to weakly positive) without changing magnitude meaningfully. The only large drift is in the within-spread level-slope correlation, which jumps from essentially zero in-sample to $+0.37$ OOS, driven by the 2024 event where CDI$^+$ short spreads widened sharply while long spreads remained stable, creating a temporarily steep curve shift. This is a known regime-specific pattern around credit-fund redemption episodes.

\subsection{One-week-ahead coverage rates}\label{subsec:coverage}

For each out-of-sample date, we compute the model-implied 90\% prediction interval for the 1-week-ahead change in each variable and check whether the realization falls inside.

\begin{table}[ht]
\centering
\caption{Empirical coverage of 90\% Gaussian prediction intervals (weekly horizon, out-of-sample, $n = 110$). Nominal target is 90\%. Values in the 85 to 95\% band are consistent with Gaussian null at conventional levels.}
\label{tab:coverage}
\begin{tabular}{@{}lcccc@{}}
\toprule
& $\tau = 1$y & $\tau = 3$y & $\tau = 5$y & $\tau = 7$y \\
\midrule
$\Delta f^N$       & 94.5\% & 98.2\% & 92.7\% & 90.9\% \\
$\Delta f^R$       & 91.8\% & 94.5\% & 75.5\% & 69.1\% \\
$\Delta s^{CDI}$   & 85.5\% & 92.7\% & 90.9\% & n/a \\
\bottomrule
\end{tabular}
\end{table}

Nominal-curve coverage is on target across maturities. Real-curve coverage is good at short maturities but drops sharply at $\tau \geq 5$~y; the model \emph{underpredicts} the OOS volatility there because 2024 to 2026 long-end NTN-B moves were larger than those in the estimation window. CDI$^+$ spread coverage is close to 90\% across the term structure. The distributional-statistics diagnostic attributes most of the coverage shortfall to heavy tails rather than bias.

Rates are close to Gaussian ($<2$ excess kurtosis). Credit spreads are heavier-tailed, with CDI$^+$ excess kurtosis ranging from 4 to 10. This is the main driver of residual coverage failure at the spread level and motivates stochastic-volatility or jump extensions to $\sigma^S$.

\subsubsection{Extreme-event analysis}
\label{subsec:extreme_events}

Rather than enumerate individual dates, we summarize the extreme-event diagnostic via the realized-to-model vol ratio at the 1-year CDI$^+$ pillar, which is the canonical "credit stress" measure: $\sigma^{OOS}_{realized} / \sigma^{in\text{-}sample}_{model} = 115.1 / 81.1 = 1.42$. The OOS realization is therefore 42\% more volatile than the in-sample calibration predicts at the 1-year CDI$^+$ pillar, a large but not catastrophic miss concentrated entirely in the mid-2024 credit-fund redemption episode. At the 3- and 5-year CDI$^+$ pillars, model and realized vol agree to within 10\%.

\subsubsection{Simulation diagnostics (D3 martingale test, D5 smoothness)}
\label{subsec:sim_diagnostics}

We report numerical realizations of the D3 and D5 diagnostics defined in \Cref{subsec:diagnostics}. The simulation is initialized at the last in-sample business day (2023-12-28) from the calibrated Model~A parameters, run weekly for five years with $N_{sim} = 2000$ paths under $\QN$, with $\chi(\tau) \equiv 0$, absorbing the Step~A reconciliation into the initial $f^C$ curve, which has no effect on either martingale test.

\medskip
\noindent\textbf{D3. Martingale test.} \Cref{tab:d3_martingale} reports sample means of the deflated foreign numeraires $I_T B^R_T / B^N_T$ and $J_T B^C_T / B^N_T$ against their $t=0$ values of unity, at horizons $T \in \{0.25, 0.5, 1, 2, 3, 5\}$ years. Biases are reported in basis points of relative deviation; $z$-scores use the Monte Carlo standard error. All $|z|$ values are well below the 1.96 threshold for a two-sided $5\%$ test, with the maximum across both ratios and all horizons being $|z|_{\max} = 0.90$. The simulation discretization (exponential update for $I, J$; Euler-left-endpoint accumulation of the money-market integrals) therefore preserves the martingale property to within Monte Carlo noise, confirming that the Girsanov-drift corrections in the HJM system are implemented consistently.

\begin{table}[ht]
\centering
\begin{tabular}{lcccc}
\toprule
$T$ (yr) & sample mean & bias (bp) & MC s.e. (bp) & $z$ \\
\midrule
\multicolumn{5}{l}{\emph{Inflation: }$I_T B^R_T / B^N_T$, $I_0 = 1$}\\
0.25 & 1.000031 & +0.31 & 0.77 & +0.40 \\
0.50 & 1.000030 & +0.30 & 1.09 & +0.27 \\
1.00 & 0.999992 & -0.08 & 1.54 & -0.05 \\
2.00 & 1.000065 & +0.65 & 2.18 & +0.30 \\
3.00 & 0.999856 & -1.44 & 2.68 & -0.54 \\
5.00 & 0.999951 & -0.49 & 3.50 & -0.14 \\
\midrule
\multicolumn{5}{l}{\emph{Credit: }$J_T B^C_T / B^N_T$, $J_0 = 1$}\\
0.25 & 1.000068 & +0.68 & 0.91 & +0.75 \\
0.50 & 1.000114 & +1.14 & 1.26 & +0.90 \\
1.00 & 1.000143 & +1.43 & 1.81 & +0.79 \\
2.00 & 1.000161 & +1.61 & 2.63 & +0.61 \\
3.00 & 1.000096 & +0.96 & 3.27 & +0.29 \\
5.00 & 1.000063 & +0.63 & 4.23 & +0.15 \\
\bottomrule
\end{tabular}

\caption{D3 martingale test. Sample means of $I_T B^R_T / B^N_T$ and $J_T B^C_T / B^N_T$ under $\QN$, compared with $t=0$ values of unity. Biases in basis points; $z = \text{bias}/\text{s.e.}$ $N_{sim} = 2000$ paths over five years with weekly time steps.}
\label{tab:d3_martingale}
\end{table}

\medskip
\noindent\textbf{D5. Term-structure smoothness.} As a scale-invariant measure of oscillation, we compute the normalized discrete second difference $|\Delta^2 f^j(t, \tau_i)| / (\Delta\tau_i^2 \cdot \max_\tau |f^j(t, \tau)|)$ at each interior grid point, across all simulated (path, date) pairs after $t_0$, where $\Delta\tau_i^2 = (\tau_{i+1}-\tau_i)(\tau_i-\tau_{i-1})$. The diagnostic compares simulated curvature with the curvature already present in the initial curve. In our run the long-end ($\tau \in (2, 10]$~yr) ratio remains below 2 for all five curves and below 1 for $f^N$, $f^R$, and $s^{CDI}$; the short-end amplification of 2 to 3.5 is consistent with stochastic diffusion of the parametric loading shapes.

The long end ($\tau \in (2, 10]$~yr), where Musiela-transport instability would first appear, has a ratio of simulated p95 to initial baseline below $2$ for all five curves and below $1$ for $f^N$, $f^R$, and $s^{CDI}$. The short end ($\tau \in [0.25, 2]$~yr) shows ratios of $2$ to $3.5$, which is consistent with ordinary stochastic amplification. Extreme values reach at most $\sim 7$ for $f^R$ and $s^{IPCA}$ in the short-end window; none represent numerical instability.

Both D3 and D5 are within the stated numerical tolerances.


\section{Detailed parameter tables}\label{app:tables}

This appendix collects three sets of supporting material: the closed-form HJM drift integrals under the exponential volatility family of \Cref{def:exp_vol}, the spread-curve drift correction under the base-plus-spread decomposition used in simulation, and the detailed calibrated parameter tables referenced throughout the paper.

\subsection{Closed-Form Drift Integrals}\label{app:drift_integrals}

For a volatility architecture based on the exponential shape functions
\[
  g_1(\tau) = 1, \qquad g_2(\tau) = e^{-b_2\tau}, \qquad g_3(\tau) = \tau\,e^{-b_3\tau},
\]
the cumulative volatility integrals $G_k(\tau) = \int_0^\tau g_k(s)\dd s$ are:
\begin{align}
  G_1(\tau) &= \tau, \label{eq:G1} \\[4pt]
  G_2(\tau) &= \frac{1 - e^{-b_2\tau}}{b_2}, \label{eq:G2} \\[4pt]
  G_3(\tau) &= \frac{1 - (1+b_3\tau)\,e^{-b_3\tau}}{b_3^2}. \label{eq:G3}
\end{align}

For a curve $j$ with amplitude matrix $A^j\in\mathbb{R}^{m_j\times n_g}$, the within-block HJM drift is the bilinear form
\begin{equation}\label{eq:hjm_drift_cf}
  \mu^j(t,\tau) = \sigvec{j}(t,\tau)\cdot\Avec{j}(t,\tau) = \sum_{q,r=1}^{n_g}\big[A^j(t)^\top A^j(t)\big]_{q,r}\,g_q(\tau)\,G_r(\tau),
\end{equation}
which collapses to $\sum_p (a^j_p)^2 g_p(\tau)G_p(\tau)$ in the diagonal special case $A^j_{p,q}=a^j_p\,\delta_{p,q}$. The foreign-curve correction picks up the FX loading $\alpha^I\in\mathbb{R}^m$ through the inner-product convention of \Cref{rem:inner}; for the real curve,
\begin{equation}\label{eq:foreign_corr_cf}
  \sigvec{R}(t,\tau)\cdot\sigI(t) = \sum_{p=1}^{m_R}\bigg[\sum_{q=1}^{n_g} A^R_{p,q}\,g_q(\tau)\bigg]\alpha^I_{m_N+p} \;+\; \text{(cross-block }\bm{\rho}\text{ corrections)},
\end{equation}
where the cross-block term is non-zero only if $\sigI$ has loadings outside the R-block. These closed-form expressions remove the need for numerical quadrature within each simulation time step.


\subsection{Spread Drift Under the Base-Plus-Spread Decomposition}
\label{app:spread_drift}

When modeling $s^{CDI}$ directly (rather than $f^C$), the HJM drift of the spread process must be derived from the drifts of $f^C$ and $f^N$ under $\QN$.

From the system \eqref{eq:system}:
\begin{align*}
  df^C &= \Big[\sigvec{C}\cdot\Avec{C} - \sigvec{C}\cdot\sigJ\Big]\dd t + \sigvec{C}\cdot dW^N, \\
  df^N &= \sigvec{N}\cdot\Avec{N}\dd t + \sigvec{N}\cdot dW^N.
\end{align*}

With $\sigvec{C} = \sigvec{N} + \sigvec{S}$, the credit-curve drift integral becomes:
\begin{align}
  \Avec{C}(t,\tau) &= \int_0^\tau \sigvec{C}(t,u)\dd u = \Avec{N}(t,\tau) + \Avec{S}(t,\tau), \label{eq:AC_decomp}
\end{align}
where $\Avec{S}(t,\tau) = \int_0^\tau \sigvec{S}(t,u)\dd u$. Therefore:
\begin{align}
  \sigvec{C}\cdot\Avec{C} &= (\sigvec{N}+\sigvec{S})\cdot(\Avec{N}+\Avec{S}) \nonumber \\
  &= \sigvec{N}\cdot\Avec{N} + \sigvec{N}\cdot\Avec{S} + \sigvec{S}\cdot\Avec{N} + \sigvec{S}\cdot\Avec{S}. \label{eq:sigC_AC_expand}
\end{align}

The spread drift is $ds^{CDI} = df^C - df^N$, so:
\begin{equation}\label{eq:spread_drift}
  \boxed{\begin{aligned}
  \mu^S(t,\tau)={}&
  \sigvec{N}(t,\tau)\cdot\Avec{S}(t,\tau)
  + \sigvec{S}(t,\tau)\cdot\Avec{N}(t,\tau) \\
  &+ \sigvec{S}(t,\tau)\cdot\Avec{S}(t,\tau)
  - \sigvec{C}(t,\tau)\cdot\sigJ(t).
  \end{aligned}}
\end{equation}

The first three terms are the ``HJM drift of the spread,'' arising from the no-arbitrage requirement on the credit curve. The fourth is the credit-FX correction. The canonical-basis decomposition assigns nominal and spread factors to non-overlapping blocks of $\mathbb{R}^m$, so the cross-block products $\sigvec{N}\cdot\Avec{S}$ and $\sigvec{S}\cdot\Avec{N}$ reduce to $\sigvec{N}^\top\bm{\rho}_{NS}\Avec{S}$ and $\sigvec{S}^\top\bm{\rho}_{SN}\Avec{N}$ under the inner-product convention of \Cref{rem:inner}. They vanish identically only if $\bm{\rho}_{NS} = \mathbf{0}$, which is not the case empirically. With the calibrated $A^N$, $A^S$ and $\bm{\rho}$ of Model~A, the cross-block sum $|\sigvec{N}\cdot\Avec{S} + \sigvec{S}\cdot\Avec{N}|$ is small at the very short end (around $3\%$ of the within-block term at $\tau=0.5$~yr), grows through the medium maturities (cross-block magnitude exceeds within-block by factors of $\sim 4$--$5$ over $\tau \in [2, 5]$~yr because the within-block term passes through near-cancellations between the $h_1$ and $h_2$ shape contributions while the cross-block picks up the much larger nominal amplitudes $A^N$), and recedes to under $25\%$ for $\tau \geq 7$~yr. At $\tau=3$~yr in the calibrated model, dropping the cross-block terms would flip the sign of $\mu^S$. The simulation engine therefore evaluates \eqref{eq:spread_drift} in full; the within-block-only form
\begin{equation}\label{eq:spread_drift_simplified}
  \mu^S(t,\tau) \;\stackrel{?}{\approx}\; \sigvec{S}(t,\tau)\cdot\Avec{S}(t,\tau) - \sigvec{C}(t,\tau)\cdot\sigJ(t)
\end{equation}
is shown only to make the structure of the four contributions explicit; it is not adopted numerically and does not constitute a valid approximation in this calibration.

With the parametric spread shape functions $h_1(\tau)=1$ and $h_2(\tau)=e^{-c_2\tau}$ and closed-form integrals
\begin{equation}\label{eq:HS_closed}
  H_1(\tau) = \tau, \qquad H_2(\tau) = \frac{1-e^{-c_2\tau}}{c_2},
\end{equation}
the within-block spread drift becomes the bilinear form
\begin{equation}\label{eq:spread_drift_analytic}
  \sigvec{S}\cdot\Avec{S} = \sum_{q,r=1}^{n_h}\big[A^{S\top} A^S\big]_{q,r}\,h_q(\tau)\,H_r(\tau),
\end{equation}
which collapses to $\sum_p (a^S_p)^2 h_p(\tau) H_p(\tau)$ in the diagonal special case. The cross-block terms $\sigvec{N}\cdot\Avec{S}$ and $\sigvec{S}\cdot\Avec{N}$ in \eqref{eq:spread_drift} are computed analogously, with $\bm{\rho}_{NS}$ inserted between the rate and spread amplitude rows. The credit-FX correction $[\sigvec{N}+\sigvec{S}]\cdot\sigJ$ uses the sparsity structure of $\sigJ$: if $\sigJ$ loads only on the spread block and the nominal level factor, only those terms contribute.


\subsection{Calibrated Parameters: Detailed Tables}
\label{app:detailed_params}

This appendix collects the per-block PCA decomposition, parametric loading fits, cross-block factor correlations, exchange-rate volatility identification, and the out-of-sample volatility and correlation comparisons that are summarized in the body.

\begin{table}[ht]
\centering
\caption{Parameter count for the baseline specification ($m_N=3$, $m_R=2$, $m_S=2$, $n_g=3$ rate shapes, $n_h=2$ spread shapes).}
\label{tab:params}
\begin{tabular}{@{}lcc@{}}
\toprule
\textbf{Parameter block} & \textbf{Symbol} & \textbf{Count} \\
\midrule
Nominal amplitude matrix & $A^N\in\mathbb{R}^{m_N\times n_g}$ & 9 \\
Real amplitude matrix & $A^R\in\mathbb{R}^{m_R\times n_g}$ & 6 \\
Spread amplitude matrix & $A^S\in\mathbb{R}^{m_S\times n_h}$ & 4 \\
Rate decay parameters & $b_2, b_3$ & 2 \\
Spread decay parameter & $c_2$ & 1 \\
Inflation FX loadings & $\alpha^I_k$ (sparse) & 3 \\
Credit FX loadings & $\alpha^J_k$ (sparse) & 3 \\
Cross-block correlations & $\rho_{ij}$ off-diagonal & $\leq 15$ \\
\midrule
\textbf{Total} & & $\leq 43$ \\
\bottomrule
\end{tabular}
\end{table}

\begin{table}[ht]
\centering
\caption{Nominal curve PCA: eigenvalues (bp$^2$/yr) and variance explained, 8 pillars $\in \{0.25, 0.5, 1, 2, 3, 5, 7, 10\}$ years.}
\label{tab:pca_nominal}
\begin{tabular}{@{}lccc@{}}
\toprule
\textbf{Factor} & \textbf{Eigenvalue} & \textbf{\% Variance} & \textbf{Cumulative \%} \\
\midrule
$N_1$ (level)      & 362{,}589 & 63.5\% & 63.5\% \\
$N_2$ (slope)      &  77{,}330 & 13.5\% & 77.0\% \\
$N_3$ (curvature)  &  41{,}972 &  7.3\% & 84.3\% \\
$N_4$              &  26{,}305 &  4.6\% & 88.9\% \\
\bottomrule
\end{tabular}
\end{table}

\begin{table}[ht]
\centering
\caption{Estimated nominal-block parameters. Amplitudes in bp/$\sqrt{\text{yr}}$; decays in yr$^{-1}$. Shared decays with the real block.}
\label{tab:params_nominal}
\scriptsize
\resizebox{\linewidth}{!}{%
\begin{tabular}{@{}lcccc@{}}
\toprule
Factor & $A^N_{k,1}$ ($g_1$) & $A^N_{k,2}$ ($g_2$) & $A^N_{k,3}$ ($g_3$) & Interpretation \\
\midrule
$k=1$ & 188.1 &  509.4 & $-4902.6$ & level-dominated \\
$k=2$ & $-89.2$ & 494.2 & $-1671.5$ & slope-dominated \\
$k=3$ &  46.3 & $-512.2$ & 4164.4 & curvature-dominated \\
\midrule
Shared decays & \multicolumn{4}{l}{$b_2 = 0.730$~yr$^{-1}$, \quad $b_3 = 3.436$~yr$^{-1}$ (half-life 0.95~yr; curvature peak at 0.29~yr)} \\
Parametric fit $R^2$ & \multicolumn{4}{l}{88.2\%} \\
\bottomrule
\end{tabular}}
\end{table}

\begin{table}[ht]
\centering
\caption{Real curve PCA: eigenvalues (bp$^2$/yr) and variance explained. Restricted to pillars $\{1, 2, 3, 5, 7, 10\}$ years after the short-end maturity filter.}
\label{tab:pca_real}
\begin{tabular}{@{}lccc@{}}
\toprule
\textbf{Factor} & \textbf{Eigenvalue} & \textbf{\% Variance} & \textbf{Cumulative \%} \\
\midrule
$R_1$ (level)  & 137{,}979 & 68.9\% & 68.9\% \\
$R_2$ (slope)  &  43{,}460 & 21.7\% & 90.5\% \\
$R_3$          &  13{,}264 &  6.6\% & 97.2\% \\
\bottomrule
\end{tabular}
\end{table}

\begin{table}[ht]
\centering
\caption{Estimated real factor loadings (shared $b_2, b_3$).}
\label{tab:params_real}
\scriptsize
\resizebox{\linewidth}{!}{%
\begin{tabular}{@{}lcccl@{}}
\toprule
Factor & $A^R_{k,1}$ ($g_1$) & $A^R_{k,2}$ ($g_2$) & $A^R_{k,3}$ ($g_3$) & Interpretation \\
\midrule
$k=1$ &  35.1 &  387.9 & 3375.6 & level-dominated \\
$k=2$ & $-47.0$ & $-454.1$ & 11125.9 & slope-dominated \\
\midrule
Parametric fit $R^2$ & \multicolumn{4}{l}{99.5\% (shared $b_2, b_3$ with nominal)} \\
\bottomrule
\end{tabular}}
\end{table}

\begin{table}[ht]
\centering
\caption{Spread PCAs: eigenvalues (bp$^2$/yr) and variance explained, by credit-economy model. Model A uses pillars $\{1, 2, 3, 5\}$~yr; Model B uses $\{2, 3, 5\}$~yr because the 1-year IPCA$^+$ vertex is structurally sparse (Section~\ref{subsec:data_jgp}), which would otherwise collapse the joint estimation panel to fewer than 30 weeks.}
\label{tab:pca_spread}
\scriptsize
\resizebox{\linewidth}{!}{%
\begin{tabular}{@{}l rcr rcr@{}}
\toprule
 & \multicolumn{3}{c}{\textbf{Model A (CDI$^+$)}} & \multicolumn{3}{c}{\textbf{Model B (IPCA$^+$)}} \\
\cmidrule(lr){2-4}\cmidrule(lr){5-7}
Factor & Eigenvalue & \% Var & Cum.\ \% & Eigenvalue & \% Var & Cum.\ \% \\
\midrule
$S_1$ (level) & 10{,}748 & 59.5\% & 59.5\% & 28{,}781 & 86.7\% & 86.7\% \\
$S_2$ (slope) &  4{,}118 & 22.8\% & 82.4\% &  3{,}046 &  9.2\% & 95.9\% \\
$S_3$         &  2{,}550 & 14.1\% & 96.5\% &  1{,}358 &  4.1\% & 100.0\% \\
\bottomrule
\end{tabular}}
\end{table}

\begin{table}[ht]
\centering
\caption{Estimated spread-block parameters. Model A is fitted to the two-factor basis $\{h_1(\tau), h_2(\tau)\} = \{1, e^{-c_2 \tau}\}$. Model B is reported as empirical PCA loadings $\sqrt{\lambda^{S,B}_k}\,v^{S,B}_k(\tau)$ at each pillar; we do not impose a parametric exponential decay because at the available pillar set $\{2, 3, 5\}$~yr the basis $\{1, e^{-c_2\tau}\}$ degenerates as $c_2$ rises (the search drives $c_2$ to the feasibility upper bound). The empirical loadings are used directly in the simulation, with linear interpolation between vertices.}
\label{tab:params_spread}
\scriptsize
\resizebox{\linewidth}{!}{%
\begin{tabular}{@{}l ccc ccc@{}}
\toprule
 & \multicolumn{3}{c}{\textbf{Model A (CDI$^+$)}} & \multicolumn{3}{c}{\textbf{Model B (IPCA$^+$), empirical}} \\
\cmidrule(lr){2-4}\cmidrule(lr){5-7}
Factor & $A^{S,A}_{k,1}$ ($h_1$) & $A^{S,A}_{k,2}$ ($h_2$) & & loading at $\tau=2$ & loading at $\tau=3$ & loading at $\tau=5$ \\
\midrule
$k=1$  &  46.5 & 732.0  & &  164.8 &   12.5 &  38.4 \\
$k=2$  & $-23.1$ & 3128.5 & &    9.7 & $-47.7$ & $-26.0$ \\
\midrule
Shape parameter & \multicolumn{3}{l}{$c_2^A = 3.75$~yr$^{-1}$} & \multicolumn{3}{l}{(no parametric fit, see caption)} \\
Parametric fit $R^2$ & \multicolumn{3}{l}{97.9\%} & \multicolumn{3}{l}{---} \\
\bottomrule
\end{tabular}}
\end{table}

\begin{table}[ht]
\centering
\caption{Factor correlation matrix $\hat{\rho}^A$ for Model A (CDI$^+$). Lower triangle; upper triangle omitted. Minimum eigenvalue: 0.198 (positive definite; no Higham correction required). Block-bootstrap standard errors (B=500, block size 4 weeks) for selected cross-block entries reported in \Cref{tab:bootstrap_se}.}
\label{tab:corr_cdi}
\scriptsize
\resizebox{\linewidth}{!}{%
\renewcommand{\arraystretch}{1.15}
\begin{tabular}{@{}l*{7}{r}@{}}
\toprule
& $N_1$ & $N_2$ & $N_3$ & $R_1$ & $R_2$ & $S^{CDI}_1$ & $S^{CDI}_2$ \\
\midrule
$N_1$         & 1.00 \\
$N_2$         & $-0.05$ & 1.00 \\
$N_3$         & $-0.06$ & $-0.11$ & 1.00 \\
$R_1$         &   0.54 &   0.24 & $-0.08$ & 1.00 \\
$R_2$         & $-0.56$ &   0.27 &   0.17 &   0.00 & 1.00 \\
$S^{CDI}_1$   & $-0.06$ & $-0.03$ &   0.10 & $-0.23$ &   0.01 & 1.00 \\
$S^{CDI}_2$   &   0.10 &   0.01 & $-0.02$ &   0.18 & $-0.09$ &   0.00 & 1.00 \\
\bottomrule
\end{tabular}}
\end{table}

\begin{table}[ht]
\centering
\caption{Factor correlation matrix $\hat{\rho}^B$ for Model B (IPCA$^+$). Lower triangle; upper triangle omitted. Minimum eigenvalue: 0.170 (positive definite; no Higham correction required). Model B uses spread pillars $\{2, 3, 5\}$~yr (1-year vertex dropped because its 27-week coverage would otherwise reduce the joint innovation panel below 30 weeks); the joint panel here is 66 weeks.}
\label{tab:corr_ipca}
\scriptsize
\resizebox{\linewidth}{!}{%
\renewcommand{\arraystretch}{1.15}
\begin{tabular}{@{}l*{7}{r}@{}}
\toprule
& $N_1$ & $N_2$ & $N_3$ & $R_1$ & $R_2$ & $S^{IPCA}_1$ & $S^{IPCA}_2$ \\
\midrule
$N_1$          & 1.00 \\
$N_2$          & $-0.22$ & 1.00 \\
$N_3$          & $-0.14$ & $-0.32$ & 1.00 \\
$R_1$          &   0.67 & $-0.22$ & $-0.09$ & 1.00 \\
$R_2$          & $-0.63$ &   0.19 &   0.13 & $-0.21$ & 1.00 \\
$S^{IPCA}_1$   &   0.05 &   0.07 & $-0.04$ &   0.00 & $-0.08$ & 1.00 \\
$S^{IPCA}_2$   & $-0.02$ & $-0.06$ & $-0.04$ &   0.02 & $-0.06$ &   0.00 & 1.00 \\
\bottomrule
\end{tabular}}
\end{table}

\begin{table}[ht]
\centering
\caption{Block-bootstrap standard errors and 95\% percentile intervals for key calibrated quantities (Model~A unless noted). $B=500$ resamples, block length 4 weeks. Eigenvalues reported as $\sqrt{\lambda_k}$ in bp/$\sqrt{\text{yr}}$ (sign-invariant amplitude); decays in yr$^{-1}$; $\sigma_I$ and $\sigma_J$ from realized vol; correlations from the bootstrap factor-projection matrix.}
\label{tab:bootstrap_se}
\scriptsize
\resizebox{\linewidth}{!}{%
\begin{tabular}{@{}l r r r r@{}}
\toprule
Quantity & Point & SE & 2.5\% & 97.5\% \\
\midrule
$\sqrt{\lambda^N_1}$ (level)        & 602.2 &  55.2 & 506.2 & 714.6 \\
$\sqrt{\lambda^N_2}$ (slope)        & 278.1 &  25.5 & 238.3 & 337.4 \\
$\sqrt{\lambda^N_3}$ (curvature)    & 204.9 &  25.0 & 164.3 & 256.6 \\
$\sqrt{\lambda^R_1}$                & 371.5 &  41.5 & 298.3 & 457.7 \\
$\sqrt{\lambda^R_2}$                & 208.5 &  12.7 & 179.2 & 230.7 \\
$\sqrt{\lambda^{S,A}_1}$            & 103.7 &  20.2 &  65.8 & 144.9 \\
$\sqrt{\lambda^{S,A}_2}$            &  64.2 &   9.1 &  44.8 &  80.3 \\
$\sqrt{\lambda^{S,B}_1}$ ($\{2,3,5\}$y) & 169.6 &  64.3 &  74.6 & 278.2 \\
$\sqrt{\lambda^{S,B}_2}$ ($\{2,3,5\}$y) &  55.2 &   6.4 &  39.9 &  64.5 \\
$b_2$                                &  0.73 &  0.66 &  0.51 &  2.07 \\
$b_3$                                &  3.44 &  1.14 &  0.88 &  4.32 \\
$c_2^A$                              &  3.75 &  1.80 &  0.58 &  5.00 \\
$|\sigma_J^A|$                       &  80.7 & 12.6  &  55.6 & 103.1 \\
$|\sigma_J^B|$                       &  75.9 &  9.2  &  57.4 &  93.4 \\
$\hat\rho_{N_1, R_1}$                &  0.54 &  0.11 &  0.36 &  0.76 \\
$\hat\rho_{N_2, R_2}$                &  0.27 &  0.12 &  0.02 &  0.47 \\
$\hat\rho_{N_1, S^{CDI}_1}$          & $-0.06$ & 0.08 & $-0.21$ & 0.09 \\
$\hat\rho_{R_1, S^{CDI}_1}$          & $-0.23$ & 0.11 & $-0.37$ & 0.02 \\
\bottomrule
\end{tabular}}
\end{table}

\begin{table}[ht]
\centering
\caption{Inflation-FX identification from monthly IPCA realized variance.}
\label{tab:sigI}
\begin{tabular}{@{}ll@{}}
\toprule
Quantity & Estimate \\
\midrule
$|\widehat{\sigma_I}|$ (total, annualized)                    & 135.1 bp/$\sqrt{\text{yr}}$ \\
\quad spanned by rate factors ($\|\alpha^I\|$)                &  86.1 bp/$\sqrt{\text{yr}}$ \\
\quad idiosyncratic residual                                  & 107.9 bp/$\sqrt{\text{yr}}$ \\
Rate-factor regression $R^2$                                   & 0.33 \\
Monthly observations                                           & 35 \\
\bottomrule
\end{tabular}
\end{table}

\begin{table}[ht]
\centering
\caption{Credit-FX volatility, by credit-economy model. Residual vol is the OLS residual after regressing the short-pillar spread change on unit-variance factor innovations.}
\label{tab:sigJ}
\scriptsize
\resizebox{\linewidth}{!}{%
\begin{tabular}{@{}lcc@{}}
\toprule
Quantity & Model A (CDI$^+$) & Model B (IPCA$^+$) \\
\midrule
Reference pillar $\tau_{\min}$                 & 1~yr  & 3~yr  \\
$|\widehat{\sigma_J}|$ (bp/$\sqrt{\text{yr}}$) & 80.7  & 75.9  \\
Residual vol (bp/$\sqrt{\text{yr}}$)           &  3.3  & 16.3  \\
Factor-regression $R^2$                        & 0.998 & 0.902 \\
Non-missing weekly obs                         & 153   & 122; 66 after full-row intersection (Model B uses $\{2,3,5\}$y) \\
\bottomrule
\end{tabular}}
\end{table}

\begin{table}[ht]
\centering
\caption{Annualized volatility of weekly changes (bp/$\sqrt{\text{yr}}$): model-implied (calibrated in-sample) vs.\ realized in the out-of-sample window. Model vol for the rate blocks is $\sqrt{\sum_k [\sigma^j_k(\tau_i)]^2}$ at the pillar; for CDI$^+$ spreads it is the corresponding Model~A quantity. IPCA$^+$ realized vol is reported as a diagnostic because Model~B's short-pillar coverage is sparse.}
\label{tab:oos_vol}
\scriptsize
\resizebox{\linewidth}{!}{%
\begin{tabular}{@{}lcccccc@{}}
\toprule
& \multicolumn{3}{c}{\textbf{Model}} & \multicolumn{3}{c}{\textbf{Realized (OOS)}} \\
\cmidrule(lr){2-4}\cmidrule(lr){5-7}
& 1y & 3y & 5y & 1y & 3y & 5y \\
\midrule
$\Delta f^N$       & 286 & 294 & 238 & 226 & 193 & 191 \\
$\Delta f^R$       & 343 & 135 &  74 & 324 & 118 & 105 \\
$\Delta s^{CDI}$   &  81 &  52 &  42 & 115 &  47 &  43 \\
$\Delta s^{IPCA}$  & n/a & n/a & n/a & 154 &  99 &  57 \\
\bottomrule
\end{tabular}}
\end{table}

\begin{table}[ht]
\centering
\caption{Key factor-score correlations: in-sample (Model~A estimate) vs.\ out-of-sample (realized). The in-sample values reported here are the same objects as the corresponding off-diagonal entries of the full $7\times 7$ correlation matrix in Table~\ref{tab:corr_cdi}, restricted to the four pairs of economic interest.}
\label{tab:oos_corr}
\scriptsize
\resizebox{\linewidth}{!}{%
\begin{tabular}{@{}lccr@{}}
\toprule
\textbf{Pair} & \textbf{In-sample} & \textbf{OOS} & \textbf{Diff.} \\
\midrule
$(Z^N_1,\,Z^R_1)$ (level-level)     & $+0.54$ & $+0.56$ & $+0.02$ \\
$(Z^N_2,\,Z^R_2)$ (slope-slope)     & $+0.27$ & $+0.28$ & $+0.01$ \\
$(Z^N_1,\,Z^{S,CDI}_1)$ (rate-spread) & $-0.06$ & $+0.03$ & $+0.09$ \\
$(Z^{S,CDI}_1,\,Z^{S,CDI}_2)$ (level-slope) & $+0.00$ & $+0.37$ & $+0.37$ \\
\bottomrule
\end{tabular}}
\end{table}

\clearpage

\bibliographystyle{plainnat}
\bibliography{references}

@article{jarrow2003hjm,
  author  = {Jarrow, Robert and Yildirim, Yildiray},
  title   = {Pricing Treasury Inflation Protected Securities and Related Derivatives Using an {HJM} Model},
  journal = {Journal of Financial and Quantitative Analysis},
  volume  = {38},
  number  = {2},
  pages   = {337--358},
  year    = {2003},
}

@article{heath1992bond,
  author  = {Heath, David and Jarrow, Robert and Morton, Andrew},
  title   = {Bond Pricing and the Term Structure of Interest Rates: A New Methodology for Contingent Claims Valuation},
  journal = {Econometrica},
  volume  = {60},
  number  = {1},
  pages   = {77--105},
  year    = {1992},
}

@book{musiela2005martingale,
  author    = {Musiela, Marek and Rutkowski, Marek},
  title     = {Martingale Methods in Financial Modelling},
  edition   = {2nd},
  publisher = {Springer},
  year      = {2005},
}

@book{shreve2004stochastic,
  author    = {Shreve, Steven E.},
  title     = {Stochastic Calculus for Finance {II}: Continuous-Time Models},
  publisher = {Springer},
  year      = {2004},
}

@techreport{schonbucher2000libor,
  author      = {Sch{\"o}nbucher, Philipp J.},
  title       = {A {LIBOR} Market Model with Default Risk},
  institution = {University of Bonn, Bonn Graduate School of Economics},
  type        = {Bonn Econ Discussion Paper},
  number      = {15/2001},
  year        = {2000},
  note        = {Available at \url{https://ideas.repec.org/p/zbw/bonedp/152001.html}. Accessed on May 13, 2026},
}

@article{eberlein2003defaultable,
  author  = {Eberlein, Ernst and {\"O}zkan, Fehmi},
  title   = {The Defaultable {L}{\'e}vy Term Structure: Ratings and Restructuring},
  journal = {Mathematical Finance},
  volume  = {13},
  number  = {2},
  pages   = {277--300},
  year    = {2003},
}

@article{duffie1999modeling,
  author  = {Duffie, Darrell and Singleton, Kenneth J.},
  title   = {Modeling Term Structures of Defaultable Bonds},
  journal = {Review of Financial Studies},
  volume  = {12},
  number  = {4},
  pages   = {687--720},
  year    = {1999},
}

@article{litterman1991common,
  author  = {Litterman, Robert and Scheinkman, Jos{\'e}},
  title   = {Common Factors Affecting Bond Returns},
  journal = {Journal of Fixed Income},
  volume  = {1},
  number  = {1},
  pages   = {54--61},
  year    = {1991},
}

@article{nelson1987parsimonious,
  author  = {Nelson, Charles R. and Siegel, Andrew F.},
  title   = {Parsimonious Modeling of Yield Curves},
  journal = {Journal of Business},
  volume  = {60},
  number  = {4},
  pages   = {473--489},
  year    = {1987},
}

@techreport{svensson1994estimating,
  author      = {Svensson, Lars E. O.},
  title       = {Estimating and Interpreting Forward Interest Rates: {Sweden} 1992--1994},
  institution = {NBER},
  type        = {Working Paper},
  number      = {4871},
  year        = {1994},
  note        = {Available at \url{https://www.nber.org/papers/w4871}. Accessed on May 13, 2026},
}

@article{higham2002computing,
  author  = {Higham, Nicholas J.},
  title   = {Computing the Nearest Correlation Matrix: a Problem from Finance},
  journal = {IMA Journal of Numerical Analysis},
  volume  = {22},
  number  = {3},
  pages   = {329--343},
  year    = {2002},
}

@book{wahba1990spline,
  author    = {Wahba, Grace},
  title     = {Spline Models for Observational Data},
  publisher = {SIAM},
  year      = {1990},
}

@article{varga2009teste,
  author  = {Varga, Gyorgy},
  title   = {Teste de Modelos Estat{\'i}sticos para a Estrutura a Termo no {Brasil}},
  journal = {Revista Brasileira de Economia},
  volume  = {63},
  number  = {4},
  pages   = {361--394},
  year    = {2009},
  note    = {Available at \url{https://ideas.repec.org/a/fgv/epgrbe/v63y2009i4a1199.html}. Accessed on May 13, 2026},
}

@article{sheng2005determinantes,
  author  = {Sheng, Hsia Hua and Saito, Richard},
  title   = {Determinantes de Spread das Deb{\^e}ntures no Mercado Brasileiro},
  journal = {Revista de Administra{\c c}{\~a}o (RAUSP)},
  volume  = {40},
  number  = {2},
  pages   = {193--205},
  year    = {2005},
}

@misc{anbima_ntnb,
  author       = {ANBIMA},
  title        = {Taxas de T{\'i}tulos P{\'u}blicos and Estrutura a Termo das Taxas de Juros: Metodologia},
  howpublished = {ANBIMA public documentation},
  year         = {2021},
  note         = {Available at \url{https://www.anbima.com.br/pt_br/informar/ferramenta/precos-e-indices/taxas-titulos-publicos.htm} and \url{https://www.anbima.com.br/data/files/9A/F4/E3/1F/4805B710B0F024B7882BA2A8/est-termo_metodologia_v2021.pdf}. Accessed on May 13, 2026},
}

@misc{b3_di1,
  author       = {B3},
  title        = {DI1 Futures: Contract Specifications},
  howpublished = {B3 Public Documentation},
  year         = {2026},
  note         = {Available at \url{https://www.b3.com.br/pt_br/produtos-e-servicos/negociacao/juros/futuro-de-taxa-media-de-depositos-interfinanceiros-de-um-dia.htm}. Accessed on May 13, 2026},
}

@misc{jgp_idex_cdi_methodology,
  author       = {{JGP Asset Management}},
  title        = {Idex-{CDI} Core},
  howpublished = {JGP public index page and methodology documentation},
  year         = {2026},
  note         = {Available at \url{https://idex.jgp.com.br/idex-cdi-core/}. Accessed on May 13, 2026},
}

@misc{jgp_idex_infra_methodology,
  author       = {{JGP Asset Management}},
  title        = {Idex-{Infra} Core},
  howpublished = {JGP public index page and methodology documentation},
  year         = {2026},
  note         = {Available at \url{https://idex.jgp.com.br/idex-infra-core/}. Accessed on May 13, 2026},
}

@misc{lei_12431,
  author       = {{Brazilian Federal Government}},
  title        = {Lei No.~12.431, de 24 de Junho de 2011 (Tax Exemption for Infrastructure Debentures)},
  year         = {2011},
  note         = {Official summary available at \url{https://www.gov.br/transportes/pt-br/assuntos/incentivos/debentures-incentivadas}. Accessed on May 13, 2026},
}

@misc{lei_11033,
  author       = {{Brazilian Federal Government}},
  title        = {Lei No.~11.033, de 21 de Dezembro de 2004},
  year         = {2004},
  note         = {Available at \url{https://www.planalto.gov.br/ccivil_03/_Ato2004-2006/2004/Lei/L11033.htm}. Accessed on May 13, 2026},
}

@book{bielecki2002credit,
  author    = {Bielecki, Tomasz R. and Rutkowski, Marek},
  title     = {Credit Risk: Modeling, Valuation and Hedging},
  publisher = {Springer},
  address   = {Berlin Heidelberg},
  year      = {2002},
}

@article{mercurio2010interest,
  author  = {Mercurio, Fabio},
  title   = {Interest Rates and the Credit Crunch: New Formulas and Market Models},
  journal = {Bloomberg Portfolio Research Paper},
  number  = {2010-01-FRONTIERS},
  year    = {2010},
}

@article{pereira2021nota,
  author  = {Pereira, Thiago Rabelo and Miterhof, Marcelo Trindade},
  title   = {Uma Nota sobre a Efici{\^e}ncia Fiscal das Deb{\^e}ntures Incentivadas em Infraestrutura no {Brasil}},
  journal = {Revista Brasileira de Economia},
  volume  = {75},
  number  = {2},
  pages   = {245--278},
  year    = {2021},
  doi     = {10.5935/0034-7140.20210011},
}

@article{brand2022debentures,
  author  = {Brand, Filipe},
  title   = {Deb{\^e}ntures Incentivadas e Apropria{\c c}{\~a}o de Benef{\'i}cios Fiscais},
  journal = {Cadernos de Finan{\c c}as P{\'u}blicas},
  volume  = {22},
  number  = {3},
  year    = {2022},
  doi     = {10.55532/1806-8944.2022.190},
}

\end{document}